\documentclass[a4paper,12pt]{article}
\usepackage[margin=1.2in,bottom=1in,top=1in]{geometry}
\usepackage{todonotes}
\setuptodonotes{inline}
\usepackage{pdflscape}
\usepackage{arydshln}
\usepackage{graphicx} 
\usepackage{ragged2e} 
\usepackage{comment} 
\usepackage{enumitem}
\usepackage{caption}
\usepackage{hyperref}
\hypersetup{
colorlinks=true,
urlcolor= blue,
citecolor=blue,
linkcolor= black
}
\usepackage{braket}
\usepackage{framed}
\usepackage{soul}
\definecolor{blue}{RGB}{33, 118, 199}

\definecolor{green}{RGB}{0, 128, 0}

\definecolor{red}{RGB}{230, 0, 20}

\usepackage{tikz-cd}
\usepackage{amsmath}
\usepackage{amssymb}
\usepackage{amsthm}
\newtheorem{definition}{Definition}
\newtheorem{prop}{Proposition}
\newtheorem{theorem}{Theorem}
\newtheorem{corollary}{Corollary}
\newtheorem{lemma}{Lemma}
\usepackage{comment}
\usepackage{dsfont}
\usepackage{cancel}
\usepackage{longtable}
\usepackage{empheq}
\usepackage{multirow}
\newcommand{\N}{\mathcal{N}}

\newcommand{\R}{\mathbb{R}}
\newcommand{\Z}{\mathbb{Z}}

\usepackage[T1]{fontenc}
\usepackage{fancyhdr}

\begin{document}

\begin{titlepage}
\vspace*{15mm}
\begin{center}
\begin{FlushLeft}
{ \Large \sffamily  \textbf{The Quasicrystalline String Landscape} }

\rule{\linewidth}{1pt}
\vspace*{5mm}
{
Zihni Kaan Baykara$^1$, Houri-Christina Tarazi{$^{2,3}$}, Cumrun Vafa$^1$\\}
\vspace{.6cm}
{ \itshape $^1$ Department of Physics, Harvard University, Cambridge, Massachusetts, USA}\par
\vspace{.2cm}
{\itshape $^2$ Enrico Fermi Institute \& Kadanoff Center for Theoretical Physics, University of Chicago, Chicago, IL 60637, USA}\par
\vspace{.2cm}
{\itshape $^3$ Kavli Institute for Cosmological Physics, University of Chicago, Chicago, IL 60637, USA}\par
\vspace{.2cm}\par
\vspace{-.3cm}

\vspace{1cm}
\begin{abstract}
In this work we investigate a largely unexplored non-geometric corner of the string landscape: the quasicrystalline orbifolds.  These exist at special points of the Narain moduli leading to frozen moduli and large quantum symmetries. 
Here we complete the classification and construction of quasicrystalline Narain lattices and use this to explore supersymmetric compactifications in $4\leq D\leq 6$ and with $4\leq Q\leq 16$ supercharges, leading to novel theories including theories with large quantum symmetries at all points in the moduli space. 
We anticipate these constructions will have many applications and in subsequent papers we apply these techniques to construct non-geometric F-theory models as well as new non-supersymmetric tachyon free models.  Similarly these constructions can lead to constructing exotic matter representations in the string landscape.
\end{abstract}
\end{FlushLeft}
\end{center}
\end{titlepage}

\tableofcontents

\newpage

\section{Introduction}

Investigation of the string theory landscape has proven to be a fruitful endeavour for both understanding the possible string theory vacua and hence connecting to phenomenology but also to understand fundamental principles of quantum gravity as characterized by the Swampland program \cite{Vafa:2005ui}. Since string theory models seem to carry much of the workload, having a clear picture of the possible landscape seems crucial.  A special focus on exotic constructions seems natural in order to find the `boundaries' of the string landscape. For example, the largest class of supersymmetric compactifications we know  are provided by Calabi-Yau manifolds but such theories usually come with a large number of massless moduli and a "universal hypermultiplet" as a consequence of their geometric nature.  However, non-geometric constructions, and in particular
asymmetric orbifold models \cite{Narain:1986qm} have proven to be valuable allies in the investigation of finding counter-examples to such naive expectations based on geometric string landscape.

In particular, in \cite{Dabholkar:1998kv} various string islands were engineered with 16 supercharges and in \cite{Baykara:2023plc} it was demonstrated that new non-geometric theories can be constructed using these methods which showed the expectation of geometric string constructions do not hold:
The `universal hypermultiplet' ended up being not so universal after all and does not exist in some of these constructions. Another example is that the celebrated Kodaira condition, familiar in the context of F-theory constructions, were shown not to be valid in some classes of asymmetric orbifolds.

Moreover, an interesting observation of \cite{Kachru:1995wm} was that non-geometric models may be connected to geometric models by utilizing transitions that are beyond the supergravity regime of the geometric model. For example, in \cite{Baykara:2023plc} it was demonstrated that there could be small volume transition in F-theory that freeze the base moduli.

In this work we continue these investigations by considering the most extreme case of non-geometric asymmetric orbifolds called quasicrystalline orbifolds, first introduced in \cite{Harvey:1987da}. Such orbifolds are characterized by symmetries of the Narain lattice that do not descend from discrete symmetries of some bulk torus in the traditional way but correspond only to symmetries of the momentum/winding Narain lattice, where the left and right momentum lattices are quasicrystals.  Necessary features for the existence of Narain lattices with such symmetries were pointed out in \cite{Harvey:1987da} and used to construct some new orbifold models.  In this paper we also show the sufficiency of these conditions by actually constructing the quasicrystalline Narain lattice.

Quasicrystalline symmetries are well studied in the context of crystallography and have various applications in phases of matter \cite{PhysRevLett.53.2477}. In string theory, lattices and crystallography show up in the context of compactification characterizing the physical charges and spectrum. As usual they can be thought of as compactifications on some geometric torus with some Kalb-Ramond B field turned on. In this work we employ their useful orbifold symmetries to construct new supersymmetric models with a small number of moduli and exotic representations of matter fields.

Interestingly, in theories with 16 supercharges we find that all quasicrystalline compactifications in 6d are part of the geometric $K_3$ moduli space. The $K_3$ automorphisms were classified in \cite{Gaberdiel:2011fg} which we use to identify with some of these quasicrystalline points.
Additionally, some of these orbifolds will provide evidence for the existence of more 6d string islands \cite{Dabholkar:1998kv} as predicted in \cite{Fraiman:2022aik} which will be constructed in \cite{Baykara:tobe}.
Similarly, we consider compactifications to dimensions $3\leq d\leq 6$ and with $4\leq \mathcal{Q} \leq 16$ supercharges and we identify the interesting features and connections to geometric models when applicable. One particularly interesting observation in quasicrystalline compactifications is that we have large discrete symmetries which we identify as the quantum symmetries of the orbifold.

The organization of the paper is as follows. In \autoref{sec2:lattices} we  review basic properties of Narain compactifications and lattices with a focus on identifying all the possible abelian automorphisms corresponding to irreducible \textit{crystalline} and \textit{quasicrystalline} symmetries. We also provide a proof that for any given choice of such symmetry it will be possible to construct an even unimodular lattice and hence provide a consistent string theory.
In \autoref{sec3:Orbifolds} we review basic properties of orbifolds that will be used in all the constructions.
Finally, in \autoref{sec4:susy} we study particular examples of \textit{quasicrystalline orbifolds} and hence produce new lower dimensional theories with ${\cal Q }=4,8,16$ supercharges in dimensions $d=4,5,6$.

\section{Lattices and Symmetries} \label{sec2:lattices}
In this section we review string compactifications on tori $T^d$ and their symmetries, with a focus on the exotic class of quasicrystalline compactifications. Torus compactifications are characterized by an even unimodular lattice $\Gamma$ of dimension $2d$ for type II and $2d+16$ for heterotic, whose symmetries lift to symmetries of the compactification. Quasicrystalline compactifications use the fact that the lattice $\Gamma$ can have a higher dimensional symmetry that can not be accommodated in $d$ dimensions. In other words, these symmetries can not be realized crystallographically on $T^d$.

This section is an overview of the more mathematical material in \autoref{app:Narain}, where proofs for the claims made in this section can be found. More generally, readers who would like to gain familiarity with lattice theory methods will find the appendix useful.

In \autoref{sub:Narain-lattices}, we review compactifications of string theory on tori $T^d$ and their equivalent characterizations by Narain lattices $\Gamma$. In \autoref{symmetries}, we explain the correspondence between the symmetries of the lattice $\Gamma$, T-dualities and symmetries of string theory. Lastly, in \autoref{quasicomp}, we review and extend the quasicrystalline compactifications \cite{Harvey:1987da}.

\subsection{Narain lattices}\label{sub:Narain-lattices}
String theory compactified on a $d$-dimensional torus $T^d$ has ground state momenta taking values in the compact dimensions in an embedded even self-dual lattice $\mathrm{\Gamma}^{d+x;d}\subset \mathbb R^{d+x;d}$, where $x=16$ for heterotic strings and $x=0$. The embedded lattice $\Gamma^{d+x;d}$ is called the \textit{Narain lattice} \cite{Narain:1985jj, Narain:1986am}. Torus compactifications are completely characterized by the choice of the Narain lattice.

We emphasize that the Narain lattice is not only a lattice, but a choice of polarization, i.e. how it splits to left and right momenta. There is a unique even unimodular lattice $\mathrm{II}^{d+x;d}$ up to isomorphism (Prop. \ref{prop:unique-lattice}). Therefore, all Narain lattices are lattice isomorphic. What actually determines the physics of the compactification is the \textit{choice of embedding} $\mathrm{II}^{d+x;d}\hookrightarrow \Gamma^{d+x;d}\subset \mathbb{R}^{d+x;d}$.

We denote the ground state left and right moving momenta on the torus as $(p_L;p_R)\in\Gamma^{d+x;d}\subset \mathbb R^{d+x;d}$. The Narain lattice determines the mass spectrum and the choice of background fields on the torus $T^d$ for type II and heterotic strings as follows.

\textbf{Type II:}

The mass spectrum is given by
\begin{align}
    M^2_L&=N_L+{p_L^2\over 2}-\frac 1 2,\\
    M^2_R&=N_R+{p_R^2\over 2}-\frac 1 2,
\end{align}
where $N_L,N_R$ are left and right moving oscillator numbers and $(p_L;p_R)\in \Gamma^{d;d}$ are the left and right ground state momenta.

The correspondence between the Narain lattice and the background fields is \cite{Blumenhagen:2013fgp}
\begin{align}
    (p_I)_{L,R}=\sqrt{\frac{\alpha'}{2}} \left(\pi_I \pm \frac 1 {\alpha'}(G_{IJ}\mp B_{IJ}) L^J\right).
\end{align}
Here, $G_{IJ},B_{IJ}$ are the background fields on $T^d$, $\pi_I$ is the center of mass momentum along the compact dimensions, and $L^I$ is the winding length
\begin{align}
    X^I(\sigma+2\pi ,\tau) = X^I(\sigma,\tau) + 2\pi L^I,
\end{align}
with $I=1,\dots,d$.

\textbf{Heterotic:}

The mass spectrum is given by
\begin{align}
   & M^2_L=N_L+{p_L^2\over 2}-1 \\
    & M^2_R=N_R+{p_R^2\over 2}-\frac 1 2,
\end{align}
where $N_L,N_R$ are left and right moving oscillator numbers and $(p_L;p_R)\in \Gamma^{d+16;d}$ are the left and right ground state momenta.

The correspondence between the Narain lattice and the background fields is \cite{Blumenhagen:2013fgp}
\begin{align}
    (p_I)_{L,R} &=\sqrt{\frac{\alpha '}{2}} \left( \tilde{\pi}_I - \frac{1}{\alpha'}(B_{IJ}\mp G_{IJ}) L^J - \pi_A A^A_I-\frac{1}{2} A_{IA} A^A_J L^J\right)\\
    (p_A)_L &= \pi_A + A_{IA}L^I.
\end{align}
Here, $I=1,\dots, d$ ranges over the $T^d$ coordinates and $A=d+1,\dots,d+16$ ranges over the 16 internal left-moving bosonic coordinates. In addition, $\tilde{\pi}_I$ is the center of mass momentum in the $T^d$ directions, $L^J$ is the winding length, and $G_{IJ},B_{IJ},A_{IA}$ are the background metric, antisymmetric, and gauge fields. Lastly, $\pi_A$ is the momentum in the 16 internal left-moving bosonic coordinates.

\subsection{Lattice automorphisms}\label{symmetries}

In this section, we describe the correspondence between lattice automorphisms of the Narain lattice $\Gamma^{d+x;d}$ and dualities and symmetries of string theory.

The automorphism group of the unique even unimodular lattice $\mathrm{Aut}(\mathrm{II}^{d+x;d})$ (which is the same as that of the Narain lattice $\mathrm{Aut}(\Gamma^{d+x;d})$) is called the \textit{T-duality group}. The T-duality group includes the familiar T-dualities sending the radii to their inverses
\begin{align}\label{eq:T-duality}
    R \mapsto \frac{\alpha '}{R},
\end{align}
as well as more complicated T-duality actions together with discrete isometries of $T^d$. Note that the T-duality group is independent of the specific choice of the Narain lattice.

Some T-duality group elements can also be symmetries of the theory depending on the choice of Narain lattice, i.e. the embedding $\mathrm{II}^{d+x;d}\hookrightarrow \Gamma^{d+x;d}\subset \mathbb R^{d+x;d}$. In particular, if an automorphism $\theta \in \mathrm{Aut}(\Gamma^{d+x,d})\subset \mathrm{O}(d+x,d,\mathbb R)$ decomposes (by virtue of the embedding) into left and right rotations as $\theta=(\theta_L;\theta_R)$, then it acts as a symmetry on the worldsheet CFT. We call the group of such automorphisms the \textit{Narain symmetry group}
\begin{align}
    \mathrm{Sym}(\Gamma^{d+x;d}):=\mathrm{Aut}(\Gamma^{d+x;d}) \cap \left(\mathrm{O}(d+x,\mathbb R)\times \mathrm{O}(d,\mathbb R)\right).
\end{align}
These are the T-dualities that act as symmetries on the worldsheet. An example is T-duality at the self-dual radius.

Note that $\theta \in \mathrm{Aut}(\Gamma^{d+x;d})$ must be similar to an integer matrix
\begin{align}\label{eq:integer-matrix-cond}
    S\theta S^{-1}\in \mathrm{O}(d+x,d,\mathbb Z)
\end{align}
with $S$ a real matrix. This is because the symmetry $\theta$ is an automorphism of a lattice, therefore its action in the basis given by the generators of the lattice must be an integer matrix. Integrality ensures that lattice elements map to lattice elements. The necessary and sufficient condition for a symmetry $\theta$ to be similar to an integer matrix is given in \autoref{thm:decomposition}. It is equivalently stated as the following. Suppose $N$ is the smallest integer such that $\theta^N=1$. If we fix an integer $p$ that divides $N$, and consider all integers $r<p$ that are coprime with $p$ as $\gcd (r,p)=1$,\footnote{Such integers $r$ are called totatives of $p$.} then all phases $e^{2\pi i r/p}$ must appear with the same multiplicity as eigenvalues of $\theta$. 

If the symmetry $\theta$ acts on the left and right movers on the worldsheet in the same way as $\theta_L=\theta_R$, then it is called a \textit{symmetric action}, and corresponds to a geometric rotation of the target space torus $T^d$. If the actions on the left and right are unequal $\theta_L\neq \theta_R$, then it is called an \textit{asymmetric action} and there is no corresponding action on the target space coordinates.

We give an example for each. For a symmetric action example, consider a string compactification on $T^2$ with complex modulus at $\tau=e^{2\pi i/3}$. There is a geometric $\mathbb Z_3$ symmetry of the target space torus that lifts to a $\mathbb Z_3$ action on the worldsheet with $\theta_L=\theta_R=R(2\pi /3)$ in an appropriate basis, where $R$ is a rotation matrix. As an asymmetric action example, T-duality \eqref{eq:T-duality} at the self-dual radius becomes a symmetry of the worldsheet. It is an asymmetric action, for which $\theta_R=1,\theta_L=-1$ with no corresponding action on the target space coordinates.

\subsection{Quasicrystalline compactifications}\label{quasicomp}

Another useful classification for Narain symmetries $\theta$ that we now describe is whether they are crystallographic or quasicrystallographic. 

A symmetry $\theta$ is a \textit{crystallographic symmetry} if its action can be written as an integer matrix on the left and right separately up to conjugation as
\begin{align}
    \theta=(\theta_L;\theta_R), \qquad Q\theta_LQ^{-1} \in \mathrm{O}(d+x,\mathbb Z),P\theta_RP^{-1} \in \mathrm{O}(d,\mathbb Z),
\end{align}
where $Q,P$ are real matrices. A crystallographic symmetry can either be a symmetric or asymmetric action.

A symmetry $\theta$ is a \textit{quasicrystallographic symmetry} if it is similar to an integer matrix but not separately on the left and right
\begin{align}
    S\theta S^{-1}\in \mathrm{O}(d+x,d,\mathbb Z), \qquad Q\theta_L Q^{-1}\not\in \mathrm{O}(d+x,\mathbb Z),P\theta_R P^{-1}\not\in \mathrm{O}(d,\mathbb Z)
\end{align}
for any $Q,P$ real matrices. This reflects the fact that the combination of the left and right actions together are irreducible; it is not possible to consider one side without the other. A quasicrystallographic symmetry is always an asymmetric action. A string theory compactification with a quasicrystallographic symmetry acting on its Narain lattice is called a \textit{quasicrystalline compactification} and was first introduced in \cite{Harvey:1987da}.

We now give an example for each case. For a crystallographic symmetry example, we can construct a Narain lattice using the weight $\Lambda_W(\mathfrak{g})$ and root $\Lambda_R(\mathfrak{g})$ lattices of a simply laced Lie algebra $\mathfrak g$ as
\begin{equation}
    \Gamma^{d;d}(\mathfrak g)=\{p_L,p_R\in \Lambda_W(\mathfrak g)|p_L-p_R\in \Lambda_R(\mathfrak g)\}.
\end{equation}
This is the usual construction for heterotic compactifications with $\mathfrak g$ gauge enhancement. For concreteness, we can choose $\mathfrak g = A_2$ and construct $\Gamma^{2;2}$ explicitly using the gluing construction of Appendix \ref{app:glue}
\begin{align}\label{eq:A2-Narain}
    \Gamma^{2;2}(A_2) = (A_2;A_2) \cup (A_2;A_2)+\left(\frac 1 3, \frac 2 3;\frac 1 3, \frac 2 3\right) \cup (A_2;A_2)+2\left(\frac 1 3, \frac 2 3;\frac 1 3, \frac 2 3\right).
\end{align}
There are both symmetric and asymmetric crystallographic actions on the lattice. For example, the asymmetric action $(R(2\pi /3);I)$, which acts with a $\mathbb Z_3$ on the left $A_2$ and identity on the right $A_2$, is a symmetry of the Narain lattice \eqref{eq:A2-Narain}. It is asymmetric and crystallographic. As another example, $(R(2\pi/3);R(2\pi/3))$ is a symmetric and crystallographic action.

For a quasicrystalline symmetry example, we define a Narain lattice $\Gamma^{2;2}$ basis as embedded in $\mathbb R^{2,2}$ as
\begin{equation}
\begin{aligned}\label{eq:Z12-quasi}
    v_1 &= \frac{1}{\sqrt[4]{3}}\left(1,0;1,0\right)\\
    v_2 &= \frac{1}{\sqrt[4]{3}}\left(\frac{\sqrt{3}}{2},\frac 1 2;-\frac{\sqrt{3}}{2},\frac 1 2\right)\\
    v_3 &= \frac{1}{\sqrt[4]{3}}\left(\frac 1 2,\frac{\sqrt{3}}{2};\frac 1 2,-\frac{\sqrt{3}}{2}\right)\\
    v_4 &= \frac{1}{\sqrt[4]{3}}\left(0,1;0,1\right).
\end{aligned}
\end{equation}
It can be checked that the lattice generated by $v_1,\dots, v_4$ is unimodular. This lattice is constructed by choosing $\theta = (R(2\pi /12); R(2\pi 5/12))$ and taking $v_n:=\theta^{n-1}\cdot v_1$. The Narain lattice has a quasicrystallographic symmetry $\theta$ by construction. Note that $\theta$ satisfies the integer matrix condition: the totatives of $12$ are $1,5,7,11$, with all of them appearing with the same multiplicity of one as eigenvalues of $\theta$. In particular $1,11$ appear once and correspond to eigenvalues $e^{2\pi i/12},e^{2\pi i 11/12}$ of $\theta$ on the left and $5,7$ also appear once and correspond to eigenvalues $e^{2\pi i 5/12},e^{2\pi i 7/12}$ on the right.

An action $\theta$ is \textit{irreducible} if it satisfies the integer matrix condition \eqref{eq:integer-matrix-cond} minimally: the only eigenvalues of $\theta$ are $e^{2\pi i r/N}$ each with multiplicity one where $r<N$ are totatives of $N$, $\gcd (r,N)=1$. An \textit{irreducible quasicrystal} is a Narain lattice with a quasicrystalline action that is irreducible. We give a list of all possible irreducible unimodular quasicrystals in \autoref{tab:unimodular-quasi-body}.

To construct irreducible unimodular quasicrystals, we fix the order $m$ of the quasicrystalline symmetry and count the number of its totatives given by the Euler totient function $\phi(m)$ (which counts the number of integers less than $m$ that are prime to m), which gives the dimension of the Narain lattice $\Gamma^{r;s}_m$ to be constructed as $\phi(m)=r+s$. For unimodularity and existence of $\Gamma^{r;s}_m$, Corollary \autoref{cor:prime-pow} provides three necessary and sufficient conditions to check: \begin{enumerate}
    \item $r\equiv s\pmod{8}$, 
    \item $r\equiv s\equiv 0\pmod{2}$, 
    \item  $m$ is not a prime power $m\neq p^a$ or two times a prime power $m\neq 2 p^a$ for some integer $a$.
\end{enumerate} 

Then one can construct $\Gamma^{r;s}_m$ by choosing a starting vector $v\in \mathbb R^{r;s}$ and taking the span of $v_n:=\theta^{n-1}\cdot v$ for $0\leq n<\phi(n)$. This construction is along the same lines as \eqref{eq:Z12-quasi} and also described in great detail in Appendix \ref{app:construction}. 

In addition, one can construct unimodular quasicrystals that are not irreducible by gluing nonunimodular qusicrystals together, as explained in Appendix \ref{app:construction}. For example, the $\mathbb Z_5$ quasicrystal $\Gamma^{2;2}_5$ in 2d is not unimodular, but gluing two copies gives a 4d unimodular quasicrystal $\Gamma^{2;2}_5\Gamma^{2;2}_5[11]$.

We now make contact with the usual notion of quasicrystals. Crystals are structures with translation invariance and possibly rotational symmetries. Quasicrystals do not have translation invariance but still have rotational symmetries. They can be constructed by projecting down a higher dimensional crystal to a subspace at an irrational angle \cite{Kramer-Neri}. In our context, the higher dimensional crystal is the Narain lattice $\Gamma^{d+x;d}$, and the quasicrystal is obtained by projecting to the left or right movers. In \autoref{fig:quasicrystal-visual}, we show the quasicrystal structure in the $(p_L^1,p_L^2)$ plane, which is obtained by projecting the Narain lattice $\Gamma^{2;2}_{12}$ to the subspace of left movers as shown in \autoref{fig:quasicrystal-cut}.

\begin{figure}
    \centering
    \includegraphics[width=\textwidth]{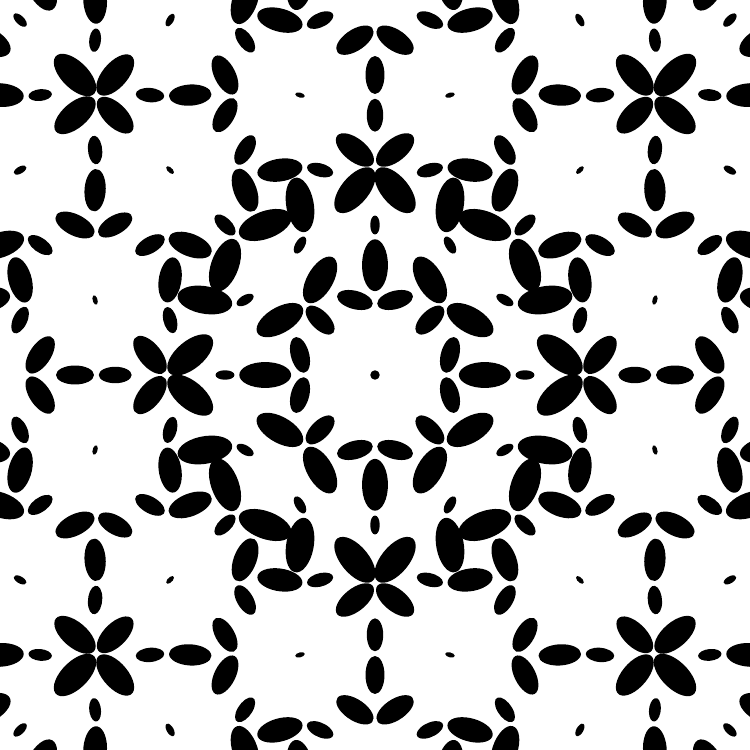}
    \caption{We give a presentation of the Narain lattice $(p_L^1,p_L^2;p_R^1,p_R^2)\in \Gamma^{2;2}_{12}$ with quasicrystalline symmetry $\mathbb Z_{12}$. The centers of the ellipses correspond to $(p_L^1,p_L^2)$, and the orientation and length of the ellipses correspond to $(p_R^1,p_R^2)$. Observe that there is no translational symmetry in the $(p_L^1,p_L^2)$ plane, but there is an overall rotational symmetry. To realize the rotational symmetry, a $\frac{2\pi}{12}$ rotation in the $(p_L^1,p_L^2)$ plane must be accompanied by a $\frac{10\pi}{12}$ rotation in the $(p_R^1,p_R^2)$ plane that acts on each ellipse around its center.}
    \label{fig:quasicrystal-visual}
\end{figure}
\begin{figure}
    \centering
    \includegraphics[width=0.5\textwidth]{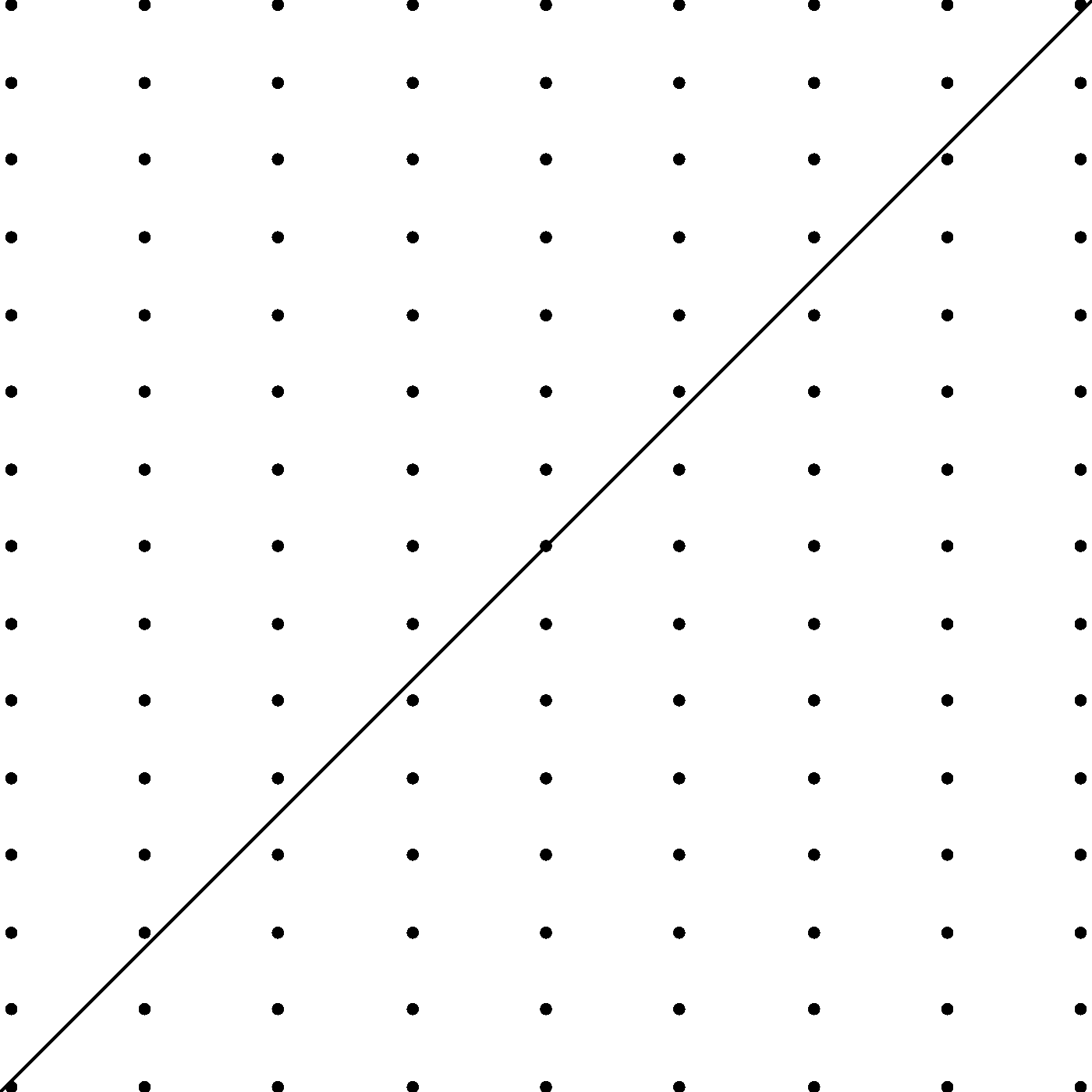}
    \caption{Quasicrystals are obtained by projecting a lattice to a subspace at special irrational angles. The figure shows the 4d unimodular lattice $\Gamma^{2;2}_{12}$ in the $(p_L^1;p_R^1)$ plane. The 2d $(p_L^1,p_L^2)$ subspace is represented here as a line that cuts the lattice at an irrational angle. Projection of the 4d lattice onto the 2d subspace produces \autoref{fig:quasicrystal-visual}.}
    \label{fig:quasicrystal-cut}
\end{figure}

\begin{table}
\begin{align*}
    \begin{array}{|c|c|c|}
    \hline
        \text{Lattice rank } \phi(m)=r+s & \text{Signature }(r;s) & \text{Symmetry order }m \\
        \hline
        4 & (2;2) & 12\\
        8 & (4;4) &15, 20, 24, 30\\
        12 & (6;6), (10;2) & 21, 28, 36, 42\\
        16 & (8;8), (12;4)&40, 48, 60\\
        20 & (14;6), (18;2) &33, 44, 66\\
        24 & (16;8), (20;4)&35, 39, 45, 52, 56, 70, 72, 78, 84, 90\\
        28 & (22;6)& \varnothing\\
        32 & (24; 8) & 51, 68, 80, 96, 102, 120\\
        \hline
    \end{array}
\end{align*}
\caption{Irreducible unimodular quasicrystals for each rank and signature. For each symmetry order $m$ and signature $(r;s)$ listed, there exists a unimodular quasicrystal $\Gamma^{r;s}_m$ with $\mathbb Z_m$ symmetry. The relevant ones are constructed explicitly in \autoref{app:construction}.}
\label{tab:unimodular-quasi-body}
\end{table}

Narain lattice data is equivalent to the background field data of the compactification. Therefore we can identify the the values of the background fields corresponding to the quasicrystalline compactifications. We use the technique explained in \autoref{app:covariant}. The $\mathbb Z_{12}$ quasicrystalline Narain lattice is spanned by \eqref{eq:Z12-quasi}. Therefore, the lattice corresponding to the spacetime torus
\begin{align}
    T^d = \mathbb R^d / 2\pi \Lambda_d
\end{align}
is spanned by
\begin{align}
    \Lambda_d &= \left\langle \sqrt{\frac{\alpha'}{2}}[(v_i)_L-(v_i)_R]\right\rangle\\
            &= \sqrt{\alpha'}\left\langle \left(\frac{\sqrt[4]{3}}{\sqrt{2}},0\right),\left(0,\frac{\sqrt[4]{3}}{\sqrt{2}}\right)\right\rangle=:\langle e_1,e_2\rangle.
\end{align}
The metric is then
\begin{align}\label{eq:fixmetric}
    [G_{ij}]=\alpha'\begin{pmatrix}
        \frac{\sqrt{3}}{2} & 0\\
        0 & \frac{\sqrt{3}}{2}
    \end{pmatrix}.
\end{align}
To find the B-field, we solve
\begin{align}
    \frac{1}{\alpha'}[B_{ij}] \mathbb Z^2 = \mathbb Z^2 \cup [G_{ij}] \Lambda_{kk}.
\end{align}
The KK momentum lattice is
\begin{align}
    \Lambda_{KK} &= \left\langle \sqrt{\frac{1}{2\alpha'}}[(v_i)_L+(v_i)_R]\right\rangle\\
            &= \frac{1}{\sqrt{\alpha'}}\left\langle \left(\frac{1}{\sqrt{2}\sqrt[4]{3}},0\right),\left(0,\frac{1}{\sqrt{2}\sqrt[4]{3}}\right)\right\rangle\\
            &= \frac{1}{\alpha' \sqrt{3}}\langle e_1,e_2\rangle.
\end{align}
Therefore in lattice basis
\begin{align}
    \Lambda_{kk} = \frac{1}{\alpha' \sqrt{3}}\left\langle (1,0),(0,1)\right\rangle
\end{align}
We see that
\begin{align}
    [G_{ij}]\Lambda_{kk}=\frac 1 2 \langle (1,0),(0,1)\rangle,
\end{align}
therefore, the B-field is given as
\begin{align}\label{eq:fixB}
    [B_{ij}] = \alpha'\begin{pmatrix}
        0 & \frac 1 2\\
        -\frac 1 2 & 0
    \end{pmatrix}.
\end{align}

\section{Orbifolds}\label{sec3:Orbifolds}
This section reviews the orbifolding methods in string theory. We emphasize asymmetric orbifolds, in which the left and right movers on the worldsheet are orbifolded by asymmetric actions. Their orbifolded spectrum usually have interesting non-geometric features. Asymmetric orbifolds of quasicrystalline compactifications provide an even more unconventional arena to search for such features. Before approaching this arena, we equip the reader with the details of orbifold techniques.

In \autoref{sub:action-symmetries}, we describe in detail how symmetries of the Narain lattice lift to actions on the string worldsheet. In \autoref{how-orbifold}, we describe how to compute the orbifold spectrum.
\subsection{Action of Narain symmetries on the worldsheet}\label{sub:action-symmetries}
The Narain lattice consists of left and right moving momenta $(p_L;p_R)$ of the string on $T^d$. By complexifying the torus coordinates according to the planes of rotation of $\theta_L, \theta_R$, we can assume that they are diagonal matrices. For $\theta_R$,
\begin{align}
    \theta_R = \mathrm{diag}(e^{2\pi i \phi^R_1},e^{-2\pi i \phi^R_1},\dots,e^{-2\pi i \phi^R_{\lfloor \frac d 2 \rfloor}})
\end{align}
in even $d$ and
\begin{align}
    \theta_R = \mathrm{diag}(e^{2\pi i \phi^R_1},e^{-2\pi i \phi^R_1},\dots,e^{-2\pi i \phi^R_{\lfloor \frac d 2 \rfloor}},\pm 1)
\end{align}
in odd $d$. Diagonal form of $\theta_L$ takes a similar form with $d$ replaced with $d+x$ and $\phi^L$ with $\phi^R$. The action of $\theta_R$ on the rotation planes is characterized by the \textit{twist vector}
\begin{align}
    \phi^R = (\phi^R_1,\dots, \phi^R_{\lfloor \frac d 2 \rfloor})
\end{align}
where entries take values between $0\leq \phi^R_i<2$. The mod $2$ value of the entry facilitates the spin uplift. Similarly, $\theta_L$ is characterized by $\phi^L$.

We also define the invariant sublattice under $g^m$ as
\begin{align}
    I(g^m):=\mathrm{Fix}_{g^m} (\Gamma^{d+x,d}) = \{p\in \Gamma^{d+x,d}\mid g^m \cdot p=p\}.
\end{align}
We let $I:=I(g)$.

One can also use shifts $v=(v_L,v_R)\in \mathbb Q\otimes \Gamma^{d+x;d}$ together with the rotations and get left-right asymmetric actions
\begin{align}
    g=(\theta_L,v_L;\theta_R,v_R)\in (\mathrm{O}(d+x) \times \mathrm{O}(d))\ltimes (\mathbb Q \otimes \Gamma^{d+x;d}).
\end{align}
We denote the projection of $v$ onto $I$ as $v^*$.

The action $g$ can be uplifted to the worldsheet CFT as $\hat{g}$, acting on the complexified oscillators as
\begin{align}
    \hat g \cdot \alpha_n^i = e^{2\pi i\phi^L_i} \alpha_n^i
\end{align}
where $i$ corresponds to a complexified torus coordinate (or possibly a real coordinate if $d$ is odd). The action on the right movers is similar. The action on the lattice modes are given as
\begin{align}
    \hat{g}\cdot \ket{p_L;p_R} = e^{2\pi i (-p_L\cdot v_L + p_R \cdot v_R)} \ket{\theta_L\cdot p_L;\theta_R\cdot p_R}.
\end{align}

The action of $\hat{g}$ on the Ramond ground state $\ket{s}=\ket{s_1,s_2,s_3,s_4},s_i=\pm \frac 1 2$ involves two points. First, an odd number of $-1$ eigenvalues are not allowed, since such an action would flip the chirality. This restricts the discussion to $\mathrm{SO}(d)$. Second, the Ramond ground state is a spacetime spinor, so rotations must be uplifted from $\mathrm{SO}$ to $\mathrm{Spin}$ groups. The uplift choice is encoded by the modulo 2 value of the twist vector entries as
\begin{align}
    \hat{g} \cdot \ket{s^{L,R}} = e^{2\pi i \phi^{L,R} \cdot s^{L,R}}\ket{s^{L,R}}.
\end{align}
It follows that the condition for supersymmetry to be preserved is
\begin{align}
    \pm \phi^{L,R}_1 \pm \phi^{L,R}_2\pm \phi^{L,R}_3\pm \phi^{L,R}_4 \equiv 0 \pmod{2}
\end{align}
for some choice of signs.\footnote{This is the condition to preserve gravitini in the untwisted sector. In general, it is possible to get gravitini from the twisted sectors of an orbifold as well. We give an example in \autoref{sec:Q=16-4d}.}
\subsection{Orbifolding procedure}\label{how-orbifold}
String theory on orbifold backgrounds was introduced in \cite{Dixon:1985jw}. Geometrically, \textit{orbifolds} can be constructed from a torus $T^d$ by quotienting by a cyclic group $\mathbb Z_N=\langle g \rangle$ generated by isometry $g$ as $T^d/\mathbb Z_n$. On the worldsheet, the orbifolding procedure amounts to relaxing the boundary conditions of the string
\begin{align}
    X^i(\tau,\sigma+2\pi) = g^n\cdot X^i(\tau,\sigma),
\end{align}
and then projecting to the invariant subspace of $\hat{g}$
\begin{align}
    \hat{g}\cdot \ket{\psi} = \ket{\psi}.
\end{align}
The states with $g^n$-twisted boundary conditions make up the \textit{$\hat{g}^n$-twisted sector}. The $n=0$ sector corresponds to the \textit{untwisted sector}.

For geometric orbifolds, the action of $g$ on the Narain lattice is left-right symmetric as $\theta_L=\theta_R$ and $v_L=v_R$. Since the left and right degrees of freedom of strings are decoupled, one can generalize the orbifolding procedure to left-right asymmetric actions $g$ with $\theta_L\neq \theta_R$ or $v_L\neq v_R$. The procedure is carried out on the worldsheet CFT in a similar fashion by relaxing the left and right boundary conditions and projecting to the invariant subspace of $\hat{g}$.
Such orbifolds have no target space interpretation and are called \textit{asymmetric orbifolds} \cite{Narain:1986qm,Narain:1990mw}.

Level matching is necessary and sufficient to ensure the consistency of the orbifolding procedure \cite{Vafa:1986wx}. In particular, for a $\mathbb Z_N$ orbifold, the energy levels on the left $E_L$ and right $E_R$ must only differ by an integer multiple of $\frac 1 N$
\begin{align}
    E_R-E_L\in \frac{\mathbb Z}{N}.
\end{align}

It is enough to check level matching for the $\hat g$-twisted sector ground state. The ground state energy in a sector twisted by twist vector $\phi$ is
\begin{align}
    (E_0)_L=\frac 1 2\sum_i \{\phi_i\} (1-\{\phi_i\})- 1
\end{align}
for bosonic,
\begin{align}
    (E_{0})_{L,R}=\frac 1 2\sum_i \{\phi_i\}-\frac 1 2
\end{align}
for supersymmetric strings, where $0\leq \{a\}<1$ is the fractional part. Additionally, if the twist is accompanied by a shift $(v_L;v_R)$, it contributes to the ground state energy on the left as $(v^*_L)^2/2$, and similarly on the right.\footnote{Strictly speaking, $(v^*)^2$ should be taken as the norm of the closest vector to the origin in $v^*+I$. In general, it can be difficult to compute, also known as the Closest Vector Problem.} Their difference is given by the indefinite norm
\begin{align}
    (v^*)^2=-(v_L^*)^2+(v_R^*)^2.
\end{align}

Therefore, level matching condition is
\begin{align}
    \frac 1 2 \left(\sum_i \{\phi^R_i\}-\sum_i \{\phi^L_i\}\right) +\frac{(v^*)^2}{2}\in \frac{\mathbb Z}{N}
\end{align}
for type II,
\begin{align}
    \frac 1 2 \left(\sum_i \{\phi^R_i\}-\sum_i \{\phi^L_i\}(1-\{\phi^L_i\})\right) +\frac{(v^*)^2}{2}+\frac 1 2\in \frac{\mathbb Z}{N}
\end{align}
for heterotic strings.

More generally, the mass in the $m$th twisted sector is given by
\begin{align}
    E_{L,R} = N_B+\frac{(r_{L,R}+m\phi_{L,R})^2}{2}+\frac{(p_{L,R}+mv_{L,R})^2}{2}+(E_0)_{L,R} - \frac 1 2
\end{align}
for the supersymmetric sides, and
\begin{align}
    E_L = N_B+\frac{(p_L+mv_L)^2}{2} + (E_0)_L-1
\end{align}
for the bosonic side. Here, $r_{L,R}$ is an $\mathrm{SO}(8)$ weight and $N_B$ is the bosonic oscillator level.

For $g^N=1$ with even $N$, an additional condition is
\begin{align}
    pg^{N/2}p=0
    \quad\quad \text{mod $2$} 
\label{Eq:consistency_even}\end{align}
for all $p\in\Gamma^{d+x,d}$.  This condition can be relaxed by doubling the order of the 
group \cite{Narain:1990mw,Harvey:2017rko}.

The number of ground states in the twisted sectors is given by 
\begin{eqnarray}
    \chi(\theta)=\sqrt{det(1-\theta)\over |I^*/I|}
\end{eqnarray}
where $I$ corresponds to the invariant lattice under $\theta$ and $I^*$ its dual. Note that as described in \cite{Narain:1986qm} $\chi(\theta) \in \Z$ for even self-dual lattices.

To compute the spectrum of the orbifolded theory, we use the partition function. The partition function in the untwisted sector is constructed using $\hat g$ insertions in the trace
\begin{align}
    Z\begin{bmatrix}
        n\\
        0
    \end{bmatrix}(\tau,\bar\tau):=\mathrm{Tr}(\hat g^n q^{L_0-c/24}\bar q^{\bar L_0-\bar c/24})
\end{align}
and projecting to $\hat{g}$ invariant states
\begin{align}
    Z[0](\tau,\bar\tau) =\mathrm{Tr}\left(\frac 1 N\sum_{n=0}^{N-1} \hat{g}^n q^{L_0-c/24}\bar q^{\bar L_0-\bar c/24}\right) = \frac 1 N \sum_{n=0}^{N-1} Z\begin{bmatrix}
        n\\
        0
    \end{bmatrix}.
\end{align}

To construct the twisted sectors, one uses \textit{modular covariance}
\begin{align}
    Z\begin{bmatrix}
        n\\
        m
    \end{bmatrix}\left(\frac{a\tau+b}{c\tau+d},\frac{a\bar\tau+b}{c\bar \tau+d}\right) = Z\begin{bmatrix}
    dn-bm\\
    am-cn
    \end{bmatrix}(\tau,\bar\tau),\qquad \begin{pmatrix}
        a & b\\
        c & d
    \end{pmatrix}\in \mathrm{SL}(2,\mathbb Z).
\end{align}
The partial trace $Z\begin{bmatrix}
    n \\ m
\end{bmatrix}$ corresponds to an insertion of $\hat{g}^n$ in the $\hat{g}^m$-twisted sector
\begin{align}
    Z\begin{bmatrix}
        n\\
        m
    \end{bmatrix} = \mathrm{Tr}_{\hat g^m}\left(\hat g^n q^{L_0-c/24}\bar q^{\bar L_0-\bar c/24}\right).
\end{align}
The $\hat g^m$-twisted sector is then constructed as
\begin{align}
    Z[m](\tau,\bar\tau) = \frac 1 N \sum_{n=0}^{N-1} Z\begin{bmatrix}
        n\\
        m
    \end{bmatrix}.
\end{align}

When the orbifold order $N$ is prime, computations are substantially simplified. This is because The $\hat{g}^m$-twisted sector partition function can be constructed by summing the orbit of $Z\begin{bmatrix}
    0\\
    m
\end{bmatrix}$ under $T$, which corresponds to simply imposing level matching. Therefore, for prime orbifolds, the twisted sector spectrum is simply given by all level-matching states. For non-prime orbifolds of order $N$, one needs to compute the non-trivial projections in addition to level-matching since not all $Z\begin{bmatrix}
    n \\
    m
\end{bmatrix}$ are in the modular orbit of $Z\begin{bmatrix}
    0\\
    m
\end{bmatrix}$.

As explained in \autoref{quasicomp}, quasicrystalline compactifications have symmetries that act asymmetrically on the left and the right, so they adopt the asymmetric orbifold machinery we have just introduced. Orbifolding by such symmetries, we obtain \textit{quasicrystalline orbifolds}. The focus of this paper is to explore the various non-geometric features of the quasicrystalline orbifold landscape.

\section{Supersymmetric models}\label{sec4:susy}

Most of the supersymmetric string constructions known involve geometric orbifolds and Calabi-Yau manifolds. However, as demonstrated in \cite{Gkountoumis:2023fym,Baykara:2023plc} it is crucial to study other corners of the string landscape as more exotic constructions could be possible that provide a more complete view of the possible landscape. For instance, in \cite{Baykara:2023plc} supersymmetric examples with 8 supercharges were constructed that have no hypermultiplets, contrary to the ordinary geometric models that always lead to at least one free hypermultiplet controlling the string coupling. Additionally, for elliptic Calabi-Yau threfolds it is known that string charges need to satisfy the Kodaira condition \cite{Kumar:2010ru} and no example outside of Calabi-Yau models was known to violate this. However, in \cite{Baykara:2023plc}  asymmetric orbifolds were shown to be powerful tools to find examples that in fact violate this condition and hence taking us away from the geometric lamppost.

As was shown in the previous section quasicrystalline orbifolds go a step further than usual asymmetric models where the orbifold symmetries are symmetries of the string lattice and not of the geometric target space torus. This leads to CFTs with no clear target space interpretation. In this section we build such models with ${\cal Q}=4,8,16$ number of supercharges in various dimensions. Additionally, orbifold quantum symmetries  will be used to identify large discrete gauge groups in our theories. 

\subsection{$\cal Q$ $=16$ Supercharges}

The string landscape with 16 supercharges is very well studied. A large class of such compactifications was studied in \cite{deBoer:2001wca, Font:2021uyw,Font:2020rsk}, isolated string islands were studied in \cite{Dabholkar:1998kv} and in particular in \cite{Eguchi:1988vra} it was shown that the compactification of IIB with $(2,0)$ supersymmetry to 6d is unique and corresponds to type IIB strings on K3. The chiral 6d $(2,0)$ theory is unique as the number of tensor multiplets is fixed to 21 by anomalies. The 6d $(1,1)$ landscape has been studied in \cite{Fraiman:2022aik} with possible gauge group enhancements corresponding to lattice embeddings in some K3 lattice. It is also well known that the K3 sigma model has well-studied orbifolds limits given by $T^4/\mathbb{Z}_{2,3,4,6}$\cite{Eguchi:1988vra} plus non-abelian orbifold limits \cite{Wendland:2000ry}. In this section we would like to study \textit{quasicrystalline} orbifolds with 16 supercharges and show that in six dimensions they are dual to special limits of type II with K3 compactification which we identify. A similar observation was made in \cite{Erler:1993zy}.

As far as quasicrystalline orbifolds are concerned we have been able to identify  4 such classes of theories with ${\cal Q}=16$ in 6d. The reasoning is as follows: We are considering asymmetric actions that act non-trivially both in the right and left sectors and hence only the type II string can give ${\cal Q}=16$ since any such action will at least break some supersymmetry. For the orbifolding action to preserve half the supersymmetries of type II, the two eigenvalues should be equal both on the left and right, so the action is given by two blocks of the same quasicrystalline action $\theta \sim C(\Phi_m)\oplus C(\Phi_m)$, where $C(\Phi_m)$ is a $4\times 4$ matrix described in \autoref{app:quasi}. Since we are considering compactifications to 6d, $m$ satisfies $\phi(m)=4$ and hence there are only 4 such quasicrystalline actions as summarized  in Table \ref{tab:all-phim} given by  $\mathbb Z_5,\mathbb Z_8,\mathbb Z_{10},\mathbb Z_{12}$.

Note that the only $(2;2)$ quasicrystal that is unimodular by itself is $\Gamma_{12}^{2,2}$. The others can be glued with another copy of themselves to give a unimodular lattice, e.g. $\Gamma_{5}^{2,2}\Gamma_5^{2,2}[11]$ constructed in \autoref{app:construction} using gluing rules summarized in \autoref{app:glue}. We can also consider a compactification of these models to 5d with a free action which will lift the twisted sectors as described in \autoref{app:free}. All these models give the same massless spectrum respectively in 6d and 5d as shown in Table \ref{tab:Q16} but twisted field carry a different discrete gauge group given by the corresponding quantum symmetry.

\renewcommand{\arraystretch}{1.5}
\begin{table}[h!]
\begin{align*}
    \begin{array}{|c|c|c|c|c|}
    \hline
        \multicolumn{5}{|c|}{\text{$Q=16$ Quasicrystalline Orbifolds}}\\
        \hline 
        \text{Dimension}  & \text{Lattice} & \text{Twist} & \text{IIA} &\text{IIB} \\
      \hline & & & & \\
         \multirow{4}{*}{6} & \Gamma^{2,2}_{5}\Gamma^{2,2}_{5}[11] & \Z_{5\ } : (1,1;2,2)/5 & \multirow{4}{*}{\begin{tabular}{@{}c@{}}$\mathcal{N}=(1,1)$ \\ $G+20V$\end{tabular}} & \multirow{4}{*}{\begin{tabular}{@{}c@{}}$\mathcal{N}=(2,0)$ \\ $G+21T$\end{tabular}}\\
          &  \Gamma^{2,2}_{8}\Gamma^{2,2}_{8}[11] & \Z_{8\ } :(1,1;3,3)/8 &  &  \\
           & \Gamma^{2,2}_{10}\Gamma^{2,2}_{10}[11] & \Z_{10} : (1,1;3,3)/10 & & \\
            & 2\Gamma^{2,2}_{12} & \Z_{12}:  (1,1;5,5)/12  & & \\ & & & & \\
           \hdashline & & &   \multicolumn{2}{|c|}{}  \\
         \multirow{4}{*}{5}  & \Gamma^{2,2}_{5}\Gamma^{2,2}_{5}[11]+\Gamma^{1,1} &\Z_{5\ } :  (1,1;2,2)/5 & \multicolumn{2}{|c|}{\multirow{4}{*}{\begin{tabular}{@{}c@{}}$\mathcal{N}=2$ \\ $G+1V$\end{tabular}}}  \\
          & \Gamma^{2,2}_{8}\Gamma^{2,2}_{8}[11]+\Gamma^{1,1} & \Z_{8\ } :  (1,1;3,3)/8 &  \multicolumn{2}{|c|}{} \\
           & \Gamma^{2,2}_{10}\Gamma^{2,2}_{10}[11]+\Gamma^{1,1} &\Z_{10} : (1,1;3,3)/10 &  \multicolumn{2}{|c|}{}\\
             & 2\Gamma^{2,2}_{12}+\Gamma^{1,1} & \Z_{12} : (1,1;5,5)/12 &  \multicolumn{2}{|c|}{} \\  & & &   \multicolumn{2}{|c|}{}  \\
           \hline 
    \end{array}
\end{align*}
\caption{Quasicrystalline orbifolds with 16 supercharges in 6d and 5d. The 5d theories are coupled with a shift in the extra $\Gamma^{1,1}$, which lifts the twisted sectors given that the circle corresponding to $\Gamma^{1,1}$ is large enough.}
\label{tab:Q16}
\end{table}

An interesting observation of \autoref{tab:Q16} is that the IIA compactification has only one tensor multiplet in the gravity multiplet and hence only one string. There are two limits of this theory: the strong and weak coupling limits. We know that there is a unique type IIA weakly coupled string theory given by the K3 sigma model. Therefore, this orbifold is a point in the K3 CFT. Additionally, due to heterotic/type II duality we expect this theory to be the strong coupling limit of some perturbative heterotic model which we would like to identify. 

In particular, every $\mathcal N=(2,2)$ SCFT with central charge $c=\bar c=6$ has a CFT elliptic genus agreeing with that of $T^4$ or K3 \cite{Eguchi:1988vra}. Those with the same elliptic genus as K3 are defined to be K3 SCFTs. 

Now we determine the exact point in the K3 SCFT moduli space that corresponds to our models.
The moduli space of the $(4,4)$ non-linear sigma model is given by 
\begin{eqnarray}
   {\cal M}_{K3}=O(\Gamma^{4,20})\textbackslash O(4,20)/O(4)\times O(20)
\end{eqnarray}
The $\Gamma^{4,20}$ is the integral cohomology of K3, which corresponds to the RR charge lattice, and $O(\Gamma^{4,20})$ is the automorphism group of the lattice. The  $O(4,20)/O(4)\times O(20)$ component specifies the choice of the NSNS fields corresponding to the metric and B-field which   parameterizes the choice of a positive definite four-dimensional subspace  in $\R^{4,20}$ determining the four left- and right-moving supercharges. The supersymmetry preserving automorphisms of the non-linear sigma model consist of those elements of $O(\Gamma^{4,20})$ that leave the four dimensional  subspace fixed. In \cite{Gaberdiel:2011fg}, for a nonsingular K3 SCFT, these automorphisms are identified with subgroups of the Conway Group $\mathrm{Co}_1$ that fix a rank 4 sublattice of the Leech lattice. In particular, the quantum symmetry $\mathcal Q$ of orbifolds, reviewed in \autoref{app:quantsym}, is such an automorphism.

In \cite{Gaberdiel:2012um}, the conjugacy classes in $\mathrm{Co}_1$ that can be quantum symmetries of torus orbifolds were determined. There is only one such conjugacy class in orders 5, 8, 10, 12, and all of them fix a rank 4 sublattice. Therefore, these K3 theories can be determined by constructing the cohomology lattice $\Gamma^{4,20}$ using the rank 20 sublattice of the Leech lattice. Such unimodular lattices were constructed in \cite{Baykara:2021ger} using the fixed-sublattice list of \cite{Hoehn:2015rsa}. The fixed sublattices of the Leech lattices are denoted as HM\# following the notation of ibid. All relevant data is provided in Table \ref{tab:K3}.\footnote{The notation for the symmetry groups follow ATLAS \cite{Wilson1985ATLASOF}: where $A:B$ denotes semidirect product, $G=A.B$ denotes $G/A\cong B$, and $5^{1+2}$ denotes the
extra-special group of order $5^3$, which is an extension of $\mathbb Z_5^2$ by a central element of
order 5.}

\begin{table}[h!]
\begin{align*}
        \begin{array}{|c|c|c|c|c|}
        \hline
            \text{Quasicrystalline orbifold} & \mathcal Q\text{ charges} & \mathrm{Co}_0\text{ class}&\Gamma^{4;20} &\text{Symmetries} \\
            \hline
            \mathbb Z_{5} & (1^5,2^5)/5 & 5C&\mathrm{HM}122 & 5^{1+2}:\mathbb Z_4\\
            \mathbb Z_{8} & (1^2,2^3,3^2,4^3)/8 &  8H& \mathrm{HM}143 &\mathbb Z_8. \mathbb Z_2^3\\
            \mathbb Z_{10} & (1,2^3,3^1,4^3,5^2)/10 & 10F& \mathrm{HM}159 & D_{20}\\
            \mathbb Z_{12} & (1,2,3^2,4^3,5,6^2)/12 & 12N& \mathrm{HM}157 & D_{24}\\
            \hline
        \end{array}
    \end{align*}
    \caption{Quasicrystalline orbifolds and the corresponding K3 surfaces. The $\mathcal Q$ charges correspond to quantum symmetry charges of the scalars. We omitted the conjugate charges for brevity. The quantum symmetry is given by a unique $\mathrm{Co}_0$ class, which constructs the cohomology lattice of the K3. Lastly, we list the full $\mathcal N=(4,4)$ symmetry preserving automorphism group.}
    \label{tab:K3}
\end{table}

There are also points in the K3 moduli space that are not quasicrystalline orbifold points, but are dual to heterotic quasicrystalline compactifications under the Type IIA K3/Heterotic $T^4$ duality. In particular, the K3 model \cite{Vafa:1989ih} obtained by the orbifold of the LG model
\begin{align}
    W=z_1^3+z_2^7+z_3^{42}
\end{align}
corresponds to a $\mathbb Z_{42}$ quasicrystal $2\Gamma^{2;10}_{42}$ on the heterotic side. To see this, we point out two disjoint $\mathbb Z_{42}$ symmetries of the theory. The first is the $\mathbb Z_{42}$ orbifold quantum symmetry in the twisted sectors. The second is a $\mathbb Z_{42}$ acting on the untwisted sector, corresponding to monomials that survive the orbifolding action with
\begin{align}
    z_1^a z_2^b z_3^c,\qquad \frac{a}{3}+\frac{b}{7}+\frac{c}{42}=1,
\end{align}
where $a\in \mathbb Z_2,b\in \mathbb Z_6,c\in \mathbb Z_{41}$. There are 10 such deformations. The $\mathbb Z_{42}$ action is given by $e^{2\pi i c/42}$ on these monomials, which are exactly the phases of the moduli of $\Gamma^{2;10}_{42}$.

Next, we consider the 5d theories in \autoref{tab:Q16}. These theories are obtained as $\Z_5,\Z_8, \Z_{10}, \Z_{12}$ freely acting orbifolds, corresponding to the 6d theories on a circle with an appropriate shift  such that the twisted sectors are lifted. These models have only two moduli and in particular only one vector multiplet. Similar, expectation was discussed in \cite{Persson:2015jka}. If one of the directions in the 2d moduli space decompactifies on a circle then these theories would correspond to  6d  string islands with no other moduli than the dilaton. Such a string island for the $\Z_5$ symmetry is known and constructed in \cite{Dabholkar:1998kv}. According to \cite{ParraDeFreitas:2022wnz}\footnote{For the cases of $\Z_5,\Z_8$ two string islands are expected differing by a choice of discrete theta angle.} we expect that also the $\Z_8, \Z_{10}, \Z_{12}$ orbifolds should also correspond to 6d string islands and hence have such a decompactification limit. The exact 6d string islands will be constructed in \cite{Baykara:tobe}.

\subsubsection*{4d}\label{sec:Q=16-4d}
We can also consider a combination of \textit{quasicrystalline} and \textit{crystalline} symmetries to construct theories with 16 supercharges in 4d.

Note that since quasicrystalline symmetries always act both on the left and right, a $4d$ quasicrystalline orbifold can have at most $8$ supercharges from the untwisted sector. For more than $8$ supercharges, one must build a model with $8$ more supercharges in the twisted sectors. An example is given in \autoref{tab:4d-N=4} which has maximal rank $22$. The extra gravitini are found in the $5$th and $10$th twisted sectors.

One could also consider the same theory on a circle with a freely acting shift that projects out the twisted sectors. One then gets a 3d theory with ${Q}=8$ and massless spectrum $G+2V$.
\begin{table}[h!] 
\begin{align*}
\resizebox{15cm}{!}{$ 
    \begin{array}{|c|c|c|c|}
    \hline
        \multicolumn{4}{|c|}{\text{4d Type IIB } \mathcal N=4\text{ Quasicrystalline Orbifolds}}\\
        \hline
        \Gamma^{6,6} & \text{Twist} &\text{Gauge group}& \text{Spectrum} \\
        \hline
         \Gamma_{5}^{2,2}\Gamma_{5}^{2,2}[11]+\Gamma(A_2)& (4/5,4/5,0;2/5,2/5,2/3)   & U(1)^{22} & G+22V  \\
        \hline
    \end{array}$}
\end{align*}
\caption{Quasicrystalline orbifolds with 16 supercharges in 4d.}
\label{tab:4d-N=4}
\end{table}

\subsection{$\cal Q$ $=8$ Supercharges}
This amount of supercharges first appears in six dimensions and therefore we study models in $D=6,5,4,3$. The largest class of these models is provided by Calabi-Yau manifolds starting from IIA/IIB, M-theory or F-theory. However, such constructions have the drawback of keeping us in the geometric lamppost. To avoid such effects we consider compactifications in non-geometric backgrounds as described in \autoref{sec3:Orbifolds} focused on orbifolding by  quasicrystalline symmetries similar to examples demonstrated in the previous section. The non-geometric nature of these models is manifest by the lack of neutral hypers in the untwisted sector with the exception of the dilaton which is dictated by the fact that we are constructing these models in the perturbative string theory.  As seen in \autoref{eq:fixB} and \autoref{eq:fixmetric} quasicrystalline symmetries exist at special points with a fixed background metric and $B-$field making manifest the non-geometric nature of the compactifications. We will study various such heterotic and type II models. The interesting features we find are the generic lack of neutral scalars in most models, the violation of the Kodaira condition in various example \cite{Baykara:2023plc} and the large generic discrete symmetries. We also comment on the connectedness of such configurations to known geometric models.

\subsubsection*{6d}
The largest known class of 6d compactifications correspond to  F-theory models on elliptic Calabi-Yau threefolds \cite{Morrison:1996na,Morrison:1996pp}. However, from the bottom up perspective there are many more theories expected to potentially exist that satisfy anomaly conditions \cite{Kumar:2010ru} and pass swampland tests \cite{Tarazi:2021duw}. For example, models with no neutral hypers or those that violate the Kodaira condition.
In \cite{Baykara:2023plc} it was shown that asymmetric orbifolds provide examples that go beyond such constructions as they are naturally non-geometric and in fact do provide examples of theories that have no neutral hypers and violate the Kodaira condition. They also provide examples of 6d theories with exotic matter that are not realizable in the geometric regime of F-theory models \cite{Baykara:Matter}. However, it is believed that such theories correspond to  stringy regions of the geometric moduli space, as for a example when the F-theory base is of stringy volume and hence they are connected to F-theory models through non-geometric transitions.

For the chiral 6d $\mathcal N=(1,0)$ supergravity, there are various anomaly cancellation conditions as imposed by the generalized Green-Schwarz mechanism. The cancellation of the gravitational anomaly condition is given by 
\begin{align}\label{anomalies}
    273 - 29 T = H_0+H_c-V.
\end{align}
where $H_0,H_c$ are the number of neutral and charged hypers respectively, $V$ is the number of vector multiplets and  $T$ the number of tensor multiplets. For perturbative heterotic theories, there is always only one tensor $T=1$ due to the absence of RR fields. There are also various gauge and mixed anomaly conditions summarised in \autoref{anomalies}. In all the theories we consider there are quantum symmetries emanating from the orbifold action as we saw in the previous section and therefore every field from the twisted sector carries some non-trivial charge under this discrete symmetry. This implies that all $H_0$ fields corresponding to neutral hypers may carry discrete gauge charge if they arise from twisted sectors.

In Table \ref{tab:Q8}, we list a collection of  ${\cal Q}=8$ quasicrystalline orbifolds in 6d arising either from the type II or the heterotic string. In fact, there are four type II quasicrystalline models with the same massless spectrum of $G+9T+8V+20H_0$, corresponding to the $\mathbb Z_5, \mathbb Z_8,\mathbb Z_{10},\mathbb Z_{12}$ quasicrystals of the previous section with an additional $(-1)^{F_L}$ twist. Similarly, to the previous section the massless spectrum can be distinguished by the discrete gauge charge that they carry from the corresponding quantum symmetry as shown in \autoref{tab:K3}. For brevity we have only listed the $\mathbb Z_{12}$ model while the rest follow from \autoref{tab:Q16}. From the geometric point of view there is an F-theory elliptic  Calabi-Yau threefold with base $dP_9$ with the same low energy matter spectrum. One could suspect that that the orbifold theories we are considering here and the Higgsed phase of theory 1 in \cite{Baykara:2023plc} are special points of this Calabi-Yau moduli space.

The next class of examples we consider all correspond to heterotic orbifolds which have non-trivial gauge groups and matter.\footnote{It is interesting to note that quasicrystalline compactifications are not gauge enhanced points, therefore the rank of the orbifold gauge group is lower than what one would have expected from the usual crystallographic orbifolds.} In this  case we consider orbifolds of orders $N=5,8,12,20,30$ and with appropriate shifts which result in non-geometric models where no neutral hypermultiplets are present similarly to \cite{Baykara:2023plc}. As discussed in \cite{Baykara:2023plc} we can consider the maximally Higgsed phases of these models and compare the spectrum to known geometric models. Since these are all perturbative heterotic models there will always be exactly one tensor multiplet and hence potentially correspond to an elliptic threefold with base $\mathbb{F}_n$ for $0\leq n\leq 12$ which we identify using the non-Higgsable clusters. The spectrum can be compared with the geometry using the correspondence:
\begin{eqnarray}
  h^{2,1}(\text{CY3})=  H_0-1, \ h^{1,1}(\text{CY3})=r+h^{1,1}(B)+2=r+T+2
\end{eqnarray}
where $B$ stands for the base of the elliptic fibration. A particularly, interesting observation is that using quasicrystalline symmetries it is simple to construct examples with large discrete symmetries corresponding to the quantum symmetries of the orbifolds.  Additionally, for each theory we check whether the Kodaira condition is satisfied as done in \cite{Baykara:2023plc}.

Since all the heterotic models seem to correspond to threefolds with base $\mathbb{F}_n$ which all have heterotic duals it would be interesting to identify the location of these theories via the F-theory/Heterotic duality.

\renewcommand{\baselinestretch}{1.5}
\begin{table}[h!] 
\begin{align*}
\hspace{-0.5cm}
\resizebox{15cm}{!}{$ 
    \begin{array}{|c|c|c|c|}
     \hline 
        \multicolumn{4}{|c|}{\text{6d ${\N}=(1,0)$ Quasicrystalline Orbifolds}}\\
        \hline                                                                
\text{Type}          & \text{IIA/B}           & \text{Het}                                                                                                                              & \text{Het}                                                \\ \hline
\Gamma^{4+x;4}       & 2\Gamma_{12}^{2;2}     & 2\Gamma_{12}^{2,2}+2\Gamma(E_8)                                                                                                         & 2\Gamma_{12}^{2;2}+2\Gamma(E_8)                           \\ \hline
\text{Symmetry} & \Z_{2}\times \Z_{12} & \Z_{12} & \Z_{12} \\\hline
\text{Twist}         & (-1)^{F_L}(1,1;5,5)/12 & (1,1;5,5)/12                                                                                                                            & (1,1;5,5)/12                                              \\ \hline
\text{Shift}         & 0                      & (1,-1, 0^6,0^8)/12                                                                                                                      & (0^{13}, 1/2, 0, 7/12) \\ \hline
\text{Gauge group}   & U(1)^{8}               & E_8\times E_7\times U(1)                                                                                                                &  E_8 \times SO(12) \times SU(2)      \times U(1)    \\ \hline
\text{Matter}        &     20 (\textbf{1}^8)                 & (\textbf{1},\textbf{56})_{(5,4,3,2,1,0)}^{(1,1,2,3,1,2)} +
(\textbf{1},\textbf{1})_{(5,4,3,2,1,0)}^{(1,1,2,3,1,2)} &    (\textbf{1},\bar{\textbf{32}},\textbf{1})_{(1,2,1)}  ^{(6,2,0)}   +(\textbf{1},\textbf{32},\textbf{1})_{(5,3,1)}^{(1,2,1)}  +(\textbf{1},{\textbf{12}},\textbf{2})_{(4,2,0)}^{(1,3,2)}  \\
& & &  +(\textbf{1},{\textbf{1}},\textbf{2})_{(7,5,9,3)}^{(5,7,1,3)} +(\textbf{1},\textbf{1},\textbf{1})_{(2,10,4,8,6)}^{(2,2,12,6,12)} 
\\ \hline
\text{Spectrum}      & G+ 9T +8V+20H_0        & G+ T +382V +626H_c                                                                                                                      & G+ T +318V +562H_c                                        \\ \hline
\text{Higgsed Phase} & dP9                    & \mathbb{F}_{12}                                                                                                                         & \mathbb{F}_{12}                                            \\ \hline
\text{Kodaira} & \text{Yes} & \text{Yes} & \text{No}\\ \hline
    \end{array} $}
\end{align*}

\begin{align*}
\hspace{-0.5cm}
\resizebox{15cm}{!}{$ 
    \begin{array}{|c|c|c|c|}
     \hline 
        \multicolumn{4}{|c|}{\text{6d ${\N}=(1,0)$ Quasicrystalline Orbifolds}}\\
        \hline 
\text{Type}          & \text{Het}                                                                                                                     & \text{Het}                                          & \text{Het}                                         \\ \hline
\Gamma^{4+x;4}       & \Gamma_{5}^{2;2}\Gamma_{5}^{2;2}[11]+\Gamma(E_8)                                                                                                   & \Gamma_{8}^{6;2}\Gamma_{8}^{6,2}[11]+\Gamma(E_8)  & \Gamma_{12}^{6;2}\Gamma_{12}^{6,2}[11]+\Gamma(E_8) \\ \hline
\text{Symmetry} & \Z_{5} & \Z_{8} & \Z_{12} \\\hline
\text{Twist}         & (1,1;2,2)/5                                                                                                                    & (1,1;3,3)/8                          & (5,5;1,1)/12                               \\ \hline
\text{Shift}         & (0^{12},2, 3, 2, 4)/5 & (0^{10},6, 3, 2, 6, 0, 2)/8                                                 &       (8,20,0,8,12,40,36,20,18,15,6,24,24,12,3,15 )/12                                            \\ \hline
\text{Gauge group}   & E_8\times SO(10)\times SU(3)\times U(1)                                                                               & E_8\times SU(4)^2\times SU(2)\times U(1)                   & SU(9)\times SO(12)\times SU(2) \times U(1)                   \\ \hline
\text{Matter}        &    (\textbf{1},\textbf{16},\textbf{1})_{(15,3,-9)}^{(1,10,5)}        +(\textbf{1},\textbf{10},\textbf{3})_{(-10,2)}^{(1,5)}                                                                                                                                 &  (\textbf{1},\textbf{6},\textbf{1},\textbf{2})_{(4,0)}^{(1,3)}   +(\textbf{1},\textbf{6},\textbf{1},\textbf{1})_{(2)}^{(6)}                                                      &       ( \textbf{9,1,2})_{(1,1)}^{(1,3)}+   ( \textbf{9,1,1})_{(2,0)}^{(2,6)}+   ( \textbf{1,1,2})_{(3,1)}^{(1,6)}+   ( \textbf{1,32,1})_{(1,0)}^{(1,1)}     \\ &  +(\textbf{1},\textbf{1},\textbf{3})_{(8,-16,-4)}^{(15,5,10)}     +(\textbf{1},\textbf{1},\textbf{1})_{(12)}^{(20)}   &        +(\textbf{1},\textbf{4},\textbf{$\bar{\mathbf{4}}$},\textbf{1})_{(4)}^{(6)}       &    
+   ( \textbf{1,12,2})_{(0)}^{(2)}+   ( \textbf{1,12,1})_{(1)}^{(4)}+   ( \textbf{1,1,2})_{(1)}^{(3)}+   ( \textbf{36,1,1})_{(0)}^{(2)} +  ( \textbf{1,1,1})_{(2)}^{(10)}      \\ \hline
\text{Spectrum}      & G+ T +302V +546H_c                                                        & G+ T +18V +262H_c                                   & G+ T +150V +394H_c                                  \\ \hline
\text{Higgsed Phase} & \mathbb{F}_{12}                                  & \mathbb{F}_{0}                                      & \mathbb{F}_{0}                                 \\ \hline \text{Kodaira} & \text{No} & \text{Yes}& \text{Yes}\\ \hline
    \end{array}  $}
\end{align*}
\caption{Quasicrystalline orbifolds with $8$ supercharges in 6d. We provide the Narain lattice $\Gamma^{4+x,4}$ where $x=0$ for type II or $x=16$ for heterotic, the twist, the shift, information about the spectrum, the potential Higgsed phases of the models and whether the model satisfies the Kodaira condition. The notation $(R_1,R_2)_{([q_1,q_2],[q_3,q_4])}^{(N_1,N_2)}$ denotes that we have $N_1$ multiplets in the tensor product representation $R_1\otimes R_2$ under two continuous gauge groups $G_1\times  G_2$ with $U(1)_1\times U(1)_2$ charge $(q_1,q_2)$ and $N_2$ multiplets in $R_1\otimes R_2$ with abelian charges $(q_3,q_4)$.}
\label{tab:Q8}
\end{table}

\renewcommand{\baselinestretch}{1.5}
\begin{table}[h!] 
\begin{align*}
\hspace{-0.5cm}
\resizebox{15cm}{!}{$ 
    \begin{array}{|c|c|c|c|}
     \hline 
        \multicolumn{3}{|c|}{\text{6d ${\N}=(1,0)$ Quasicrystalline Orbifolds}}\\
        \hline 
\text{Type}          & \text{Het}                                       & \text{Het}                                 \\ \hline
\Gamma^{4+x;4}       & \Gamma_{30}^{6;2}\Gamma_{30}^{6,2}[11]+\Gamma(E_8) & \Gamma_{30}^{6;2}\Gamma_{30}^{6,2}[11]+\Gamma(E_8) \\ \hline
\text{Symmetry} & \Z_{30} & \Z_{20}  \\\hline
\text{Twist}         & (13,7,11,13,7,11;1,1)/30                             & (3, 7, 9, 3, 7, 9 ; 1, 1)/20                         \\ \hline
\text{Shift}         & (0^5,9,0^2)/10                                                & (0^5, 10, 7, 0)        /20                                      \\ \hline
    \text{Gauge group}   & E_7\times U(1)                      & E_6\times U(1)^2                   \\ \hline
\text{Matter}   &   ( \textbf{1})_{(2,\ 3,\ 4,\ 5,\ 6,\ 7)}^{(82,64,42,28,12,4)}                                                       &           ( \textbf{1})_{([7,0],\ [3,0],\ [4,0],\ [6,0],[0,3],[1,0],[2,0],[5,0])}^{(\ 1,  \quad \ 9,\quad  \ 14, \quad  6, \ \ 20, \  \  67, \ \ 72, \ \  27 \ )}                    \\ & &  ( \textbf{1})_{( [2,-3],\ [-3,3],[5,3],[5,-3],[2,3],[3,3],[1,-3],[-4,3],[6,3])}^{(\ \   6, \quad \ \ 5, \quad \ 3, \quad \ 4, \quad 24, \ \ 18,  \quad 34, \quad 12, \quad 2 \ )}                                            \\ \hline
\text{Spectrum}      & G+ T +134V +378H_c                                  & G+ T +80V +324H_c                                   \\ \hline
\text{Higgsed Phase} & \mathbb{F}_8                                      & \mathbb{F}_6                                        \\ \hline
    \end{array} $}
\end{align*}
\caption{Quasicrystalline orbifolds with $8$ supercharges in 6d. We provide the Narain lattice $\Gamma^{4+x,4}$ where $x=0$ for type II or $x=16$ for heterotic, the twist, the shift, information about the spectrum, and the potential higgsed phases of the models. The notation is as explained in \autoref{tab:Q8}.}
\label{tab:Q8B}
\end{table}

\newpage
\subsubsection*{5d}
Quasicrystalline compactifications are possible in even dimensions as reviewed in \autoref{app:classification} but one can consider the 6d theories on a circle with a free action along the circle corresponding to  shifting. Since a shift direction opens up with the extra circle, theories that might not have satisfied level matching in 6d can get cured with a shift. Such an example is demonstrated in  \autoref{tab:Q85d} corresponding to a $\Z_{42}$ orbifold. Additionally, all heterotic theories of \autoref{tab:Q8} will have an $18$ dimensional Coulomb branch as the first example of \autoref{tab:Q85d} .  The type II theory will have a $10$ dimensional Coulomb branch similar to the second example. This is because the gauge symmetry came from the untwisted sector in all these examples and hence remains after the shift.

In particular, the $\mathbb Z_{42}$ example has minimal number of vector multiplets and a very large discrete gauge symmetry group: $\mathbb Z_2.\mathbb Z_{42}$ generated by the quantum symmetry $\mathcal Q$ and reflection $-1$. This is an isolated branch in the moduli space with a $\mathbb Z_{42}$ discrete gauge symmetry at all points in the moduli space. Because of the free action generically only massive states carry charge under this symmetry. The usual geometric intuition tells us that Calabi Yau moduli spaces in general do not have generic discrete symmetries but at special points they do (for example the moduli space of the quintic  has points with $\mathbb Z_{41}$ symmetry \cite{Greene:1998yz}). From the example with generic discrete gauge symmetries the order of the symmetry group is much smaller. Therefore, both the lack of a universal hyper and the large generic discrete symmetries make these models non-geometric. But as discussed earlier we generically expect that the non-geometric  models are connected to the geometric ones at some special points.

Lastly, we mention that just as was noticed in \cite{Gkountoumis:2023fym, Baykara:2023plc} all examples without hypers have even and bounded rank in 5d. It would be interesting to investigate non-cyclic orbifold constructions which may be more promising to provide the odd rank cases, if they exist.

\renewcommand{\baselinestretch}{1.5}
\begin{table}[h!] 
\begin{align*}
\hspace{-0.5cm}
\resizebox{16cm}{!}{$ 
    \begin{array}{|c|c|c|c|c|c|c|}
    \hline
        \multicolumn{7}{|c|}{\text{5d ${\N}=1$ Quasicrystalline Freely Acting Orbifolds}}\\
        \hline
        \text{Type} & \Gamma^{21;5} & \text{Twist} & I  &\text{Gauge group}& \text{Spectrum} \\
        \hline
        \text{Het} & 2\Gamma_{12}^{2,2}+2\Gamma(E_8)+\Gamma^{1;1} & \Z_{12}:(1,1;5,5)/12 & 2\Gamma(E_8)+\Gamma^{1;1} &  U(1)^{18} & G +18 V \\
        \text{Het} & \Gamma_{20}^{2,6}\Gamma_{20}^{2,2}[11]+\Gamma(E_8)+\Gamma^{1;1} & \Z_{20}:(3,7,9,3,7,9;1,1)/20 & \Gamma(E_8)+\Gamma^{1;1} &  U(1)^{10} & G +10V\\
        \text{Het} & 2\Gamma_{42}^{2,10}+\Gamma^{1;1} & \Z_{42}:(5,11,13,17,19,5,11,13,17,19;1,1)/42 & \Gamma^{1;1} &U(1)^2 & G +2V\\
        \hline
    \end{array}$ }
\end{align*}
\caption{Freely acting quasicrystalline orbifolds with $8$ supercharges in 5d. By doing a shift on $\Gamma^{1;1}$ with a circle of radius large enough, we lift the twisted sectors. The 5d spectrum is then given by the untwisted sector of the 6d orbifold and KK modes.}
\label{tab:Q85d}
\end{table}

\subsubsection*{4d}
In \autoref{tab:4d-N=2}, we list quasicrystalline orbifolds in 4d with $Q=8$. Given that quasicrystals act on both left and right, only type II can be used to give this amount of supercharges. We list only IIB, since the spectrum of IIA can be found by exchanging the corresponding hodge numbers. In particular,  the number of vector multiplets is given by $h_{2,1}$ and of hypermultiplets is given by $h_{1,1}+1$.

The first $\mathbb Z_{12}$ quasicrystalline orbifold has the same spectrum whether considered from IIA or IIB similar to the 
self-mirror Calabi-Yau that has the same Hodge numbers.

\begin{table}[h!] 
\begin{align*}
\resizebox{15cm}{!}{$ 
    \begin{array}{|c|c|c|c|c|}
    \hline
        \multicolumn{5}{|c|}{\text{4d Type IIB } \mathcal N=2\text{ Quasicrystalline Orbifolds}}\\
        \hline
        \Gamma^{6,6} & \text{Twist} &\text{Gauge group}& \text{Spectrum}& (h_{2,1},h_{1,1}) \\
        \hline
         3\Gamma_{12}^{2,2}& (1,1,2;5,5,10)/12   & U(1)^{11} & G+11V+12H_0 & (11,11), \cite{Candelas:2016fdy} \\
         3\Gamma_{12}^{2,2} & (1,2,3;5,10,15)/12   & U(1)^{22} & G+22V+11H_0& (22, 10),\cite{Kreuzer:2000xy}\\
         3\Gamma_{12}^{2,2} & (1,3,4;5,15,20)/12   & U(1)^{14} & G+14V+27H_0& (14, 26),\cite{Kreuzer:2000xy}\\
         3\Gamma_{12}^{2,2} & (1,4,5;5,20,1)/12   & U(1)^{29} & G+29V+6H_0& (29, 5),\cite{Kreuzer:2000xy}\\
         \Gamma_{24}^{4,4}+\Gamma_{12}^{2,2}& (1,7,8;5,11,16)/24   & U(1)^8 & G+8V+21H_0& (8, 20)\\
         \Gamma_{24}^{4,4}+\Gamma_{12}^{2,2} & (1,11,10;5,7,2)/24   &U(1)^{13} & G+ 4V+17H_0& (4, 16), \cite{Candelas:2016fdy} \\
        \hline
    \end{array}$}
\end{align*}
\caption{Quasicrystalline orbifolds with $Q=8$ supercharges in 4d. We provide the spectrum for type IIB theories. The spectrum for IIA theories can be obtained by considering the mirror $h_{1,1}\leftrightarrow h_{2,1}$. In IIB we have $h^{1,1}+1$ hyper multiplets and $h^{1,2}$ vectors. The Calabi-Yau references correspond to existing Calabi-Yau manifolds or their mirrors known in the literature that would correspond to the same low energy spectrum.}
\label{tab:4d-N=2}
\end{table}

\subsubsection*{3d}
In 3d, there is only one kind of supermultiplet---the vector and matter multiplets are indistinguishable. Additionally, since it is an odd dimension, we get an extra circle that we can shift on and  lift the twisted sectors.

\begin{table}[h!]
    \begin{align*}
    \begin{array}{|c|c|c|}
    \hline
        \multicolumn{3}{|c|}{\text{3d ${\cal N }=2$ Freely Acting Quasicrystalline Orbifolds}}\\
        \hline
        \Gamma^{7,7} & \text{Twist}  & \text{Spectrum} \\
        \hline
        2\Gamma^{2,2}_{12}+\Gamma^{2,2}(A_1^2)+\Gamma^{1,1} &  (1,1,0;5,5,6)/12   & G+2V\\
        \Gamma^{2,2}_{12}+\Gamma^{4,4}_{24}+\Gamma^{1,1} &  (1,10,11;2,5,7)/24 & G+V+H\\
        2\Gamma^{2,2}_{12}+\Gamma^{2,2} &  (1,1,12;5,5,0)/12 &  G+6V\\
        \Gamma^{6,6}_{21}+\Gamma^{1,1} &  (1,4,5;2,8,10)/21 &  G+V+H\\
        \hline
    \end{array}
\end{align*}
\caption{Freely acting quasicrystalline orbifolds with $Q=4$ supercharges. We denote the 4d uplift of these models if they were to exist in order to be more descriptive about the spectrum.}
\label{tab:3d-quasi}
\end{table}

\newpage
\subsection{$\mathcal{Q}=4$ Supercharges}
Here we consider some examples of theories with $\mathcal{N}=4$ supercharges similar to F-theory on an elliptic Calabi-Yau fourfold. One could consider also evaluate the superpotential and K\"ahler potential for such theories which is left for future work.
\subsubsection*{4d}
In \autoref{tab:4d-N=1}, we list some $Q=4$ quasicrystalline orbifolds of potential interest with minimal amount of matter. We provide the charged spectra in appendix \ref{app:4dmatter}.

The heterotic $\mathbb Z_{24}$ quasicrystalline orbifold only has one neutral complex matter multiplet while the other matter multiplets are charged. This is the minimal number of neutral matter one could get because the complex axio-dilaton is always neutral and cannot be projected out in a perturbative string theory. In this sense, this model is minimal  in the perturbative string corner.

\begin{table}[h!]
\begin{align*}
\resizebox{16cm}{!}{$ 
    \begin{array}{|c|c|c|c|c|c|c|}
    \hline
        \multicolumn{7}{|c|}{\text{4d } \mathcal N=1\text{ Quasicrystalline Orbifolds}}\\
        \hline
        \text{Type}  & \text{Twist} & I & \text{Shift} &\text{Gauge group}& \text{4d Spectrum} \\
        \hline
        \text{IIA}  & (-1)^{F_L}(2,5,7;1,10,11)/24 & 0 & 0 &U(1)^{3} & G+ 3V+35M_0 \\
        \text{IIB}  & (-1)^{F_L}(2,5,7;1,10,11)/24 & 0 & 0 &U(1)^{15} & G+ 15V+23M_0 \\
        \text{Het}  & (1,1,2;5,5,10)/12 & 2\Gamma(E_8) & 
            (7, 1, 8, 11 ,
            8 , 9 , 0 , 8 ,
            5 , 10 , 10 ,5,
            1 , 8 , 9 , 3)/12 & SU(4)^2\times SU(2)^3\times U(1)^7 & G+46V+459M_c+2M_0 \\
        \text{Het} & (2,5,7;1,10,11)/24 & 2\Gamma(E_8) & (19,20,6,0,10,18,17,6,20,14,14,21,23,5,21,4)/24 &SU(2)^5\times U(1)^{11} & G+ 26V + 277M_c +M_0\\
        \hline
    \end{array}$}
\end{align*}
\caption{Quasicrystalline orbifolds with $Q=4$ supercharges in 4d.A matter multiplet consists of a complex scalar and a Weyl fermion, and $M_0, M_c$ corresponds to neutral and charged matter. The 3d spectrum has one extra vector.}
\label{tab:4d-N=1}
\end{table}

\section{Conclusion and future directions}
In this work we have studied a special type of asymmetric orbifolds called quasicrystalline orbifolds that are specified by some quasi-periodicity of the Narain lattices. They correspond to irrational 2d CFT and this irrationality helps to lift massless states. We have identified all such irreducible abelian quasicrystals in 6d with 16 supercharges and have given various examples in lower dimensions and with lower amounts of supersymmetry. We expect that such symmetries may arise at special strong coupling points in geometric models and hence in a sense many of these orbifolds may be deformable to geometric models. We have also identified large quantum symmetries which correspond to large discrete symmetries of the bulk theory. Additionally, we identified four more potential string islands in 6d with 16 supercharges by constructing their circle compactification to 5d.  Their explicit construction in 6d is in progress in \cite{Baykara:tobe}.

In a accompanying work \cite{Baykara:nonsusy} we study also applications of quasicrystalline compactifications to non-supersymmetric strings and identify such models which are tachyon free and rigid in the sense of having only one neutral modulus given by the dilaton. These theories are similar but correspond to a separate class from the $O(16)\times O(16)$ non-susy string theories \cite{Alvarez-Gaume:1986ghj,Dixon:1986iz} and their compactification \cite{Ginsparg:1986wr,Fraiman:2023cpa} but they could be related in some web of dualities.

It would be interesting to classify large classes of exotic orbifold models and understand their exotic features like the violation of the Kodaira condition, no neutral hyper models etc. Such an analysis would give a more complete idea of the consistent string landscape and potentially avoid geometric lamppost effects. Additionally, it would be interesting to continue the search for non-susy orbifold models with no tachyon that may have potential phenomenological implications.

Other interesting features of orbifolds that are of particular interest include constructions of non-perturbative orbifolds and their generic consistency. A subclass of such theories in the context of F-theory is studied in \cite{Baykara:ftheory}. This work and similar works point towards a more complete picture of the boundaries of consistent quantum gravity vacua and allows us to sharpen the Swampland criteria.

\section{Acknowledgements}

We thank Yuta Hamada for participation at early stages of this work.
HCT would like to thank Harvard University and the Swampland Initiative for their hospitality during part of this work.
 The work of ZKB and CV is supported by a grant from the Simons Foundation (602883,CV), the DellaPietra Foundation, and by the NSF grant PHY-2013858.

\appendix{}

\section{Narain lattice to background conversion}\label{app:covariant}
The correspondence between the Narain lattice and the KK momenta and winding is \cite[Eq. (10.37)]{Blumenhagen:2013fgp}
\begin{align}
p^I_{L,R}=\frac 1 {\sqrt{2}}\left(\sqrt{\alpha'}p^I \pm \frac 1 {\sqrt{\alpha'}} L^I\right),
\end{align}
where $p^I$ is the KK momentum and $L^I$ is the winding length
\begin{align}
    X^I(\sigma+2\pi ,\tau) = X^I(\sigma,\tau) + 2\pi L^I.
\end{align}

Equivalently, the KK momenta and winding length in terms of $p_L,p_R$ are given by
\begin{align}
    L^I&=\sqrt{\frac{\alpha'}{2}} (p^I_L - p^I_R),\\
    p^I&=\frac{1}{\sqrt{2\alpha'}}(p^I_L+p^I_R).
\end{align}

From the winding length data, it is easy to read off the metric. The lattice of winding lengths
\begin{align}
    \Lambda_d &=\{L \mid X(\sigma+2\pi,\tau)=X(\sigma,\tau)+2\pi L\}\\
    &= \left\{\sqrt{\frac{\alpha'}{2}} (p_L - p_R) \mid (p_L;p_R)\in \Gamma^{d;d}\right\}
\end{align}
is precisely the lattice corresponding to the spacetime torus
\begin{align}
    T^d = \mathbb R^d/2\pi \Lambda_d.
\end{align}
Therefore, choosing a basis $e_i$ for $\Lambda_d$, the background metric can be written as $G_{ij}=e_i\cdot e_j$ in the lattice basis. We use capital letters $I,J$ to denote Cartesian coordinate indices and $i,j$ to denote lattice coordinate indices. One can convert back to Cartesian coordinates by $G_{IJ}= e_I^{*i}G_{ij} e^{*j}_J$, where $e^{*i}$ form a basis for the dual lattice $\Lambda_d^*$.

Determining the $B$-field is more subtle. This is because KK momentum $p^I$ actually does not generate center of mass translations when there is a nontrivial $B$-field. Intuitively, the $B$-field assigns a phase to the worldsheet. If a string with winding $L^I$ is translated, it will pick up a phase like $e^{i B_{IJ} L^I x^I}$ corresponding to the worldsheet sweeped by that translation, independent of the phase $e^{i x_I p^I}$ due to KK momentum. Therefore, the actual generator of translations depends both on the KK momentum and the winding
\begin{align}\label{eq:canonical-conj}
    \pi_I = G_{IJ} p^J +\frac 1 {\alpha'}B_{IJ}L^J.
\end{align}

Since $\pi_I$ generates translations on $T^d$, periodicity requires
\begin{align}
    \pi_I\in \Lambda_d^*.
\end{align}
Define the KK momentum lattice as
\begin{align}
    \Lambda_{KK}=\left\{ p^I=\frac{1}{\sqrt{2\alpha'}}(p^I_L+p^I_R)\mid (p_L;p_R)\in \Gamma^{d;d}\right\},
\end{align}
and in terms of the lattice basis as
\begin{align}
    \Lambda_{kk}=\{p^i = p^Ie_I^{*i}\mid p^I\in \Lambda_{KK}\}.
\end{align}

We can solve for the $B$-field by writing \eqref{eq:canonical-conj} in lattice basis
\begin{align}
    \frac 1 {\alpha'} e^{*i}_IB_{ij}e^{*j}_J L^J&= \pi_I - e_I^{*i} G_{ij}e_J^{*j}p^J\\
    \frac 1 {\alpha'} B_{ij} L^j &= \pi_I e_i^I- G_{ij} p^j.
\end{align}
Note that $L^j$ is an integer because $L^J$ span the lattice $\Lambda_d$. Similarly, $\pi_I e^I_i$ is an integer because $\pi_I$ is in the dual lattice. Writing the above in terms of lattices, we get
\begin{align}
    \frac 1 {\alpha'} [B_{ij}] \mathbb Z^d&=\mathbb Z^d \cup [G_{ij}]\Lambda_{kk}.
\end{align}
In practice, one chooses a basis $f_i$ for the right hand side, then solve for an antisymmetric $B$-field by
\begin{align}
    f_i = \frac 1 {\alpha'} \sum_j  B_{ij}.
\end{align}

\section{Classification of Narain lattice symmetries}\label{app:Narain}
In this section, we classify all possible finite order symmetries of Narain lattices.

\autoref{app:lattice-theory} sets the definitions. We first define and differentiate between three concepts that progressively build on each other: free $\mathbb Z$-modules (structures isomorphic to $\mathbb Z^n$), lattices, and Narain lattices. In \autoref{app:classification}, we classify the symmetries of the structures defined. In particular, we classify symmetries of free $\mathbb Z$-modules, and determine the conditions for such symmetries to act on lattices. We also show that any symmetry that is a symmmetry of a unimodular lattice is also a symmetry of a Narain lattice. Finally in \autoref{app:quasi}, we reintroduce quasicrystalline compactifications in terms of the mathematical language developed.

We mention the main results here to make navigation easier. \autoref{thm:decomposition} is a classification of automorphisms of free $\mathbb Z$-modules, in other words, finite order elements $\theta\in \mathrm{GL}(n,\mathbb Z)$. Corollary \autoref{cor:prime-pow} is the sufficient (and necessary \cite{Harvey:1987da}) condition for the existence of a symmetry $\theta$ to act on a unimodular lattice. Corollary \autoref{cor:lattice-aut-exists} states that any $\theta\in\mathrm{GL}(n,\mathbb Z)$ can be used to construct an even lattice, possibly nonunimodular. In \autoref{sub:sym-embed} we give an argument that any finite order symmetry of a unimodular lattice is also the symmetry of a Narain lattice.
\subsection{Lattice theory review}\label{app:lattice-theory}
\subsubsection{Abstract lattices}
\begin{definition}
    A lattice $(\Lambda,q)$ is a free $\mathbb Z$-module $\Lambda$ equipped with a quadratic form $q$
    \begin{align}
        q(nv)=n^2q(v),\qquad n\in \mathbb Z, v\in \Lambda.
    \end{align}
\end{definition}
One can use a bilinear form $\langle -,-\rangle$ (also denoted by a dot $-\cdot -$) interchangeably with the quadratic form $q$. The conversion between the two is given by the polarization identity
\begin{align}
    \langle v,w\rangle &= \frac 1 2 \left(q(v+w)-q(v)-q(w)\right),\\
    q(v) &= \langle v,v\rangle.
\end{align}
The lattice is \textit{even} if $q(v)$ is even for all $v\in \Lambda$ and \textit{integral} if $\langle v,w\rangle\in\mathbb Z$ for all $v,w\in \Lambda$. Note that evenness implies integrality. The \textit{signature} $(r;s)$ of the lattice denotes the negative $r$ and positive $s$ indices of inertia of the quadratic form $q$. Lattices with $r=0$ (resp. $s=0$) are called \textit{positive} (resp. \textit{negative}) \textit{definite}. If $r,s\neq 0$, the lattice is \textit{indefinite}.  We denote the lattice $(\Lambda,q)$ with signature $(r;s)$ as $\Lambda^{r;s}$, whereas when we refer to the underlying free $\mathbb Z$-module we use $\Lambda$.

The minimal number $n$ of generators of $\Lambda$ is called its \textit{rank}. Without the quadratic form, $\Lambda$ of rank $n$ is isomorphic to the integer lattice $\Lambda \cong \mathbb Z^n$. With the quadratic form, the right equivalence morphisms are isometries. An \textit{isometry} between two lattices $\Lambda^{r;s}$ and $\Lambda'^{r;s}$ is a module isomorphism $\psi$ that preserves the quadratic form
\begin{align}
    \psi: \Lambda &\to \Lambda',\\
    q(v)&=q'(\psi(v)).
\end{align}
Isometric lattices are denoted as $\Lambda^{r;s}\cong \Lambda'^{r;s}$. An \textit{isometric automorphism} of a lattice $\Lambda^{r;s}$ is an automorphism of $\Lambda$ that is an isometry. The group of isometric automorphisms of a lattice are denoted as
\begin{align}
    \mathrm{Aut}(\Lambda^{r;s}).
\end{align}

\begin{definition}
The \textit{dual lattice} of $\Lambda^{r;s}$ is defined by
\begin{align}
    \Lambda^*:=\{w\in \mathbb Q\otimes \Lambda \mid \langle w,v\rangle\in \mathbb Z\text{ for all }v\in\Lambda \},
\end{align}
endowed with the $\mathbb Q$-linear extension of the bilinear form $\langle -,-\rangle$.
\end{definition}

The lattice is called \textit{unimodular} if it is dual to itself. An important result in lattice theory is that indefinite even unimodular lattices are unique.
\begin{prop}\label{prop:unique-lattice}
    If $\Lambda^{r;s}$ is an indefinite even unimodular lattice, then $r\equiv s\pmod{8}$ and it is unique up to isometry
    \begin{align}
        \Lambda^{r;s} \cong E_8(\pm)^{\oplus \frac{|r-s|}{8}} \oplus U^{\oplus \min (r,s)} =:\mathrm{II}^{r;s},
        \label{evenuni}
    \end{align}
    where $U$ is the hyperbolic lattice and $E_8(\pm)$ is the $E_8$ lattice with positive $(+)$ or negative $(-)$ definite quadratic form.
\end{prop}
\subsubsection{Narain lattice}
String theory compactified on a $d$-dimensional torus $T^d$ is characterized by the embedding of the even unimodular lattice $\mathrm{II}^{d+x;d}\hookrightarrow\Gamma^{d+x;d}\subset \mathbb R^{d+x;d}$, where $x=16$ for heterotic strings and $x=0$ otherwise \cite{Narain:1985jj, Narain:1986am}. The embedded lattice $\Gamma^{d+x;d}$ is called the \textit{Narain lattice}. 

The \textit{T-duality group} is given by $\mathrm{Aut}(\mathrm{II}^{d+x,d})\cong \mathrm{Aut}(\Gamma^{d+x,d})$. Depending on the lattice embedding, some combination of T-dualities can become symmetries. In particular, an isometric automorphism $\theta \in \mathrm{Aut}(\Gamma^{d+x,d})$ that decomposes into left and right parts as $\theta=(\theta_L;\theta_R)$ acts as a symmetry on the worldsheet CFT. We call the group of such isometric automorphisms the \textit{Narain symmetry group}
\begin{align}
    \mathrm{Sym}(\Gamma^{d+x;d}):=\mathrm{Aut}(\Gamma^{d+x;d}) \cap \left(\mathrm{O}(d+x)\times \mathrm{O}(d)\right).
\end{align}
\subsection{Classification of symmetries}\label{app:classification}
Our goal is to classify the symmetries that can occur in $\mathrm{Sym}(\Gamma^{r;s})$. We will not limit our discussion to unimodular lattices, as the gluing construction in appendix \ref{app:glue} can be used to obtain unimodular lattices from nonunimodular ones. However, we limit our discussion to even lattices.

Consider the structures we have reviewed: free $\mathbb Z$-modules $\Lambda$, lattices $\Lambda^{r;s}$, embedded lattices $\Lambda^{r;s}\hookrightarrow \Gamma^{r;s}\subset \mathbb R^{r;s}$, and the automorphism structures of each.
\begin{center}
    \begin{tikzcd}
    \Lambda \arrow[r] & \Lambda^{r;s} \arrow[r] & (\Lambda^{r;s} \hookrightarrow \Gamma^{r;s}\subset\mathbb R^{r;s})\\
    \mathrm{Aut}(\Lambda)  \arrow[r, hookleftarrow] & \mathrm{Aut}(\Lambda^{r,s}) \arrow[r, hookleftarrow] & \mathrm{Sym}(\Gamma^{r;s})
\end{tikzcd}
\end{center}
Our procedure will be to start with an element $\theta\in\mathrm{Aut}(\Lambda)$, and then to construct a quadratic form $q$ and an embedding $\Lambda^{r;s}\hookrightarrow\Gamma^{r;s}$ compatible with $\theta$ so that we can pull it back to $\mathrm{Sym}(\Gamma^{r;s})$.
\subsubsection{Free $\mathbb Z$-module automorphisms}
The automorphism group of a free $\mathbb Z$-module $\Lambda$ of rank $n$ is isomorphic to the group of invertible integer matrices $\mathrm{Aut}(\Lambda)\cong\mathrm{GL}(n,\mathbb Z)$, so we use them interchangeably. The invertible integral matrices have determinant $\pm 1$. In this subsection, we will determine what finite orders are possible in this group. For an exposition on integral matrices, see \cite{newman,5e076f03-1178-3b70-95a5-ab41b5a4d470}.

We will characterize matrices $\theta\in\mathrm{GL}(n,\mathbb Z)$ using their associated polynomials.
\begin{definition}
    The characteristic polynomial of a square matrix $\theta$ is
    \begin{align}
        \chi_\theta(x) = \det (xI-\theta),
    \end{align}
\end{definition}
\begin{definition}
    The minimal polynomial of a square matrix $\theta$ is the monic polynomial $\mu_\theta(x)$ of smallest degree such that
    \begin{align}
        \mu_\theta(\theta)=0.
    \end{align}
\end{definition}
The roots of the characteristic polynomial $\chi_\theta(x)$ are the eigenvalues of $\theta$ with multiplicity. The roots of the minimal polynomial $\mu_\theta(x)$ are the eigenvalues of $\theta$ each with multiplicity $1$. Note that $\chi_\theta(\theta)=\mu_\theta(\theta)=0$ and $\mu_\theta(x)$ divides $\chi_\theta(x)$.

Conversely, one can also define a matrix starting with a polynomial.
\begin{definition}
    Given a monic irreducible polynomial
    \begin{align}
        p(x) = x^n + a_{n-1}x^{n-1} + \dots + a_{1}x+a_0,
    \end{align}
    the companion matrix $C(p)$ is defined as
    \begin{align}
        C(p) := \begin{pmatrix}
        0 & 0 & 0 & \cdots & -a_0\\
        1 & 0 & 0 & \cdots & -a_1\\
        0 & 1 & 0 & \cdots & -a_2\\
         & \vdots & & \vdots\\
        0 & 0 & 0 & \cdots & -a_{n-1}
    \end{pmatrix}.
    \end{align}
\end{definition}
By construction, the characteristic polynomial of the companion matrix $C(p)$ is
\begin{align}
    \chi_{C(p)}(x)=p(x).
\end{align}

To construct elements of order $m$ in $\mathrm{GL}(n,\mathbb R)$, it is enough to consider one rotation plane and rotate by $2\pi r/m$ where $r$ is coprime with $m$. However, such a matrix may not be integral. 

Instead, we construct an element of order $m$ in $\mathrm{GL}(n,\mathbb Z)$ with multiple rotation planes, each with a different rotation phase $2\pi r/m$ where $0<r<m$ is coprime with $m$. Such integers $r$ are called \textit{totatives} of $m$, and the number of totatives of $m$ is given by \textit{Euler's totient function} $\phi(m)$.
\begin{definition}
    The $m$th cyclotomic polynomial is
    \begin{align}
        \Phi_m(x):=\prod_{\substack{\mathrm{gcd}(r,m)=1 \\ 0<r<m}} (x-e^{2\pi i r/m}).
    \end{align}
\end{definition}
Some important properties of cyclotomic polynomials are as follows.
\begin{lemma}\label{lem:cyclotomic-props}
    Cyclotomic polynomials have the following properties:
    \begin{enumerate}
        \item $\Phi_m(x)$ is a monic polynomial of degree $\phi(m)$,
        \item $\Phi_m(x)$ is an irreducible polynomial over the integers,
        \item The irreducible decomposition of $x^m-1$ is
        \begin{align}
            x^m-1 = \prod_{d|m}\Phi_d(x).
        \end{align}
    \end{enumerate}
\end{lemma}
\begin{proof}
Properties (1) follow from the definition of cyclotomic polynomials. For proof of property (2), see \cite[p. 554]{dummit2003abstract}. Property (3) follows from (2).
\end{proof}

Given the coefficients of the $m$th cyclotomic polynomial,
\begin{align}
    \Phi_m(x) = x^{\phi(m)} + a_{\phi(m)-1} x^{\phi(m)-1} + \dots + a_1 x + a_0,
\end{align}
we define the companion matrix
\begin{align}
    C(\Phi_m) = \begin{pmatrix}
        0 & 0 & 0 & \cdots & -a_0\\
        1 & 0 & 0 & \cdots & -a_1\\
        0 & 1 & 0 & \cdots & -a_2\\
         & \vdots & & \vdots\\
        0 & 0 & 0 & \cdots & -a_{\phi(m)-1}
    \end{pmatrix}.
\end{align}
Since each eigenvalue is distinct, both the characteristic and minimal polynomials of $C(\Phi_m)$ are given by the $m$th cyclotomic polynomial $\chi_{C(\Phi_m)}=\mu_{C(\Phi_m)}=\Phi_m$.

Now we show that $C(\Phi_m)$ are the building blocks of invertible integral matrices with finite order.

\begin{theorem}\label{thm:decomposition}
    An element $\theta\in \mathrm{GL}(n,\mathbb Z)$ of order $m=p_1^{e_1}p_2^{e_2}\dots p_t^{e_t}$ with prime $p_1<p_2<\dots <p_t$ exists if and only if $\theta$ is similar over $\mathbb Q$ to a block matrix
    \begin{align}
        Q\theta Q^{-1} = \bigoplus_{j=1}^k C(\Phi_{d_j})^{\oplus \ell_j}, \qquad \text{with }\ell_j\in\mathbb N, \quad Q\in \mathrm{GL}(n,\mathbb Q), 
    \end{align}
    where each $d_j$ divides $m$ and $\mathrm{lcm}(d_1,\dots,d_k)=m$, with $n$ satisfying
    \begin{equation}
    \begin{aligned}\label{eq:rank-ineq}
        \sum_{i=1}^t (p_i-1)p_i^{e_i-1}-1 &\leq n\qquad\text{for } p_1^{e_1}=2,\\
        \sum_{i=1}^t (p_i-1)p_i^{e_i-1}&\leq n\qquad \text{otherwise}.
    \end{aligned}
    \end{equation}
\end{theorem}
\begin{proof}
We only present the decomposition part of the theorem. For the rest of the proof, see \cite[Theorem 2.7]{5e076f03-1178-3b70-95a5-ab41b5a4d470}. To denote two matrices $A,B$ that are similar over $\mathbb Q$, we write $A\sim B$.

Suppose $\theta\in \mathrm{GL}(n,\mathbb Z)$ has order $m=p_1^{e_1}p_2^{e_2}\dots p_t^{e_t}$. Let $\mu_\theta(x)$ be the minimal polynomial of $\theta$. Then $\mu_\theta(x)$ divides $x^m-1$. Let the irreducible decomposition of $\mu_\theta(x)$ be
\begin{align}
    \mu_\theta(x) = p_1(x)^{f_1}p_2(x)^{f_2}\dots p_k(x)^{f_k}.
\end{align}
We also know the irreducible decomposition of $x^m-1$ from Lemma \ref{lem:cyclotomic-props}
\begin{align}
    x^m-1=\prod_{d|m}\Phi_d(x).
\end{align}
Therefore, $p_j(x)=\Phi_{d_j}(x)$ for some $d_j|m$ with each $d_j$ distinct and $f_j=1$
\begin{align}
    \mu_\theta(x) = \prod_{j=1}^k \Phi_{d_j}(x).
\end{align}
By primary decomposition, we get that $\theta$ is similar over $\mathbb Q$ to a block matrix
\begin{align}
    \theta \sim \begin{pmatrix}
        \theta_{d_1} & 0 &\cdots  & 0\\
        0 & \theta_{d_2} &\cdots &0\\
        \vdots & \vdots  & \ddots & \vdots \\
        0 &0 &\cdots  & \theta_{d_k}
    \end{pmatrix}
\end{align}
with $\theta_{d_i}$ having the minimal polynomial $\Phi_{d_i}(x)$. This means the eigenvalues of $\theta_{d_i}$ are totatives of $d_i$, possibly with multiplicity. Since $\theta_{d_i}$ is an integral matrix, its characteristic polynomial must be
\begin{align}
    \chi_{\theta_{d_i}}(x) = \Phi_{d_i}(x)^{\ell_i}
\end{align}
for some $\ell_i \in \mathbb N$. 
This means that it is similar over $\mathbb Q$ to a block matrix of $\ell_i$ many companion matrices
\begin{align}
    \theta_{d_i} \sim C(\Phi_{d_i})^{\oplus \ell_i}.
\end{align}
\end{proof}
Note that the similarity statement is over $\mathbb Q$ and not necessarily over $\mathbb Z$. However, similarity over $\mathbb Z$ can also be determined for certain cases. Latimer-MacDuffee Theorem \cite{Latimer} states that the number of similarity classes of integral matrices with irreducible minimal polynomial $\mu(x)$ is given by the class number of $\mathbb Z[x]/(\mu(x))$.\footnote{Reducible $p(x)$ was also allowed in the original formulation of the Latimer-MacDuffee theorem with more involved machinery. However, in general, computation of class number is difficult.} For $\mu(x)=\Phi_m(x)$, the class number is $1$ for $\phi(m)<22$. This means for $\theta\in\mathrm{GL}(n,\mathbb Z)$ of order $m$ with $\mu_\theta=\Phi_m$ and $\phi(m)<22$, the automorphism $\theta$ is unique up to conjugacy. In physical terms, this means that quasicrystalline symmetries are unique up to integral transformations for $\Gamma^{r;s}$ with $r+s<22$.

\subsubsection{Lattice automorphisms}
Given a $\theta\in\mathrm{GL}(n,\mathbb Z)$, we now determine if there is a quadratic form $q$ that it fixes as $q(\theta(v))=q(v)$ for all $v\in \Lambda$. This determines if a lattice $\Lambda^{r;s}$ with $\theta \in \mathrm{Aut}(\Lambda^{r;s})$ exists.

We start with the case of unimodular lattices, where there are stringent conditions. Given a polynomial $q(x)$, define the number of eigenvalues $\lambda$ of the companion matrix $C(q)$ with $|\lambda|>1$ by $m(q)$.
\begin{theorem}[\cite{GROSS2002265,BAYERFLUCKIGER2002215,BayerFluckigerTaelman2020}]
    Let $r,s$ be non-negative integers with $r\equiv s\pmod{8}$. Let $p(x)$ be a monic irreducible polynomial and $q(x)=p(x)^n$ for $n$ a non-negative integer. Assume $q(x)$ has degree $r+s$. If
    \begin{enumerate}
        \item $q(x)$ is reciprocal, i.e. $x^{r+s}q(1/x)=q(x)$,
        \item $m(q)\leq \min (r,s)$ \text{ and } $m(q)\equiv r \equiv s \pmod{2}$,
        \item $|q(1)|,|q(-1)|,(-1)^{\frac{r+s}{2}}q(1)q(-1)$ are squares.
    \end{enumerate}
    Then there exists an even unimodular lattice $\Lambda^{r;s}$ and $\theta\in\mathrm{SO}(\Lambda^{r;s})$ with characteristic polynomial $\chi_\theta(x)=q(x)$.
\end{theorem}
The irreducible polynomials we deal with are cyclotomic polynomials. Specializing to our case, we have
\begin{corollary}\label{cor:prime-pow}
    Let $r\equiv s \equiv 0 \pmod{2}$. If \begin{enumerate}
        \item $n$ is even, or
        \item $n=1$ and $m$ is neither a prime power $p^r$ nor two times a prime power $2p^r$,
    \end{enumerate}
    then there is an even unimodular lattice $\Lambda^{r;s}$ of rank $r+s=n\phi(m)$ and $\theta\in \mathrm{SO}(\Lambda^{r;s})$ of order $m$ with characteristic polynomial $\chi_\theta(x)=\Phi_m(x)^n$.
\end{corollary}
\begin{proof}
    Cyclotomic polynomials for $m\geq 2$ are reciprocal so (1) is satisfied. (2) is satisfied by assumptions. To show the hypotheses satisfy (3), we use some identities of $\Phi_m(x)$. In particular, 
    \begin{align}
        \Phi_m(1)=\begin{cases}
            p & \text{if }m=p^r\text{ with }p\text{ prime},\\
            1 & \text{otherwise}.
        \end{cases}
    \end{align}

    Now we consider $\Phi_m(-1)$. Suppose $m$ is not a prime power. If $m$ is odd, then the identity
    \begin{align}\label{eq:oddphi}
        \Phi_{m}(-x)=\Phi_{2m}(x)
    \end{align}
    holds. Therefore for odd $m$, $\Phi_m(-1)=\Phi_{2m}(1)=1$.
    
    For even $m=2^k t$ with $t$ odd, if $k>1$ we use the identity
    \begin{align}
        \Phi_m(x)=\Phi_{2t}(x^{2^{k-1}}),
    \end{align}
    and if $k=1$ we use \eqref{eq:oddphi}. We get
    \begin{align}
        \Phi_{2^k t}(-1) = \begin{cases}
            1, & k>1\\
            \Phi_{t}(1) & \text{otherwise}.
        \end{cases}
    \end{align}
    We see that for $m=2^kt$ with $k>1$, the condition is satisfied. For $m=2t$, $\Phi_m(-1)=1$ if and only if $t$ is not a prime power.
\end{proof}
Note that this is a sufficient condition for a quasicrystalline symmetry $\theta$ to act on a unimodular lattice. This condition is exactly what was found in \cite{Harvey:1987da} to be a necessary condition. Together, we conclude that condition 2 of Corollary \autoref{cor:lattice-aut-exists} is necessary and sufficient for an irreducible quasicrystal to exist.

More generally, we can always construct an even lattice not necessarily unimodular with the desired automorphism. Since $\theta\in\mathrm{GL}(n,\mathbb Z)$ decomposes into blocks of $C(\Phi_m)$ for various $m$, it is enough to consider only one such block.
\begin{prop}\label{prop:irred-lat-exists}
    Let $r\equiv s\equiv 0\pmod{2}$. Let $m>1$ be an integer. Then there exists an even lattice $\Lambda^{r;s}$ of rank $r+s=\phi(m)$ with isometric automorphism $C(\Phi_m) \in \mathrm{Aut}(\Lambda^{r;s})$ of order $m$.
\end{prop}
\begin{proof} Finding a quadratic form $q$ fixed by $C(\Phi_m)$ is equivalent to finding a $\phi(m)\times\phi(m)$ Gram matrix $G$ with
\begin{align}\label{eq:isometry-matrix}
    C(\Phi_m)^T G C(\Phi_m)=G.
\end{align}
This is a linear system of equations for the entries of $G$. We first solve it over the reals, then rationals, and finally integers. 

We solve the linear equation in reals. Note that $v_i^T\otimes  v_j$ for $v_i$ eigenbasis vectors of $C(\Phi_m)$ form a basis for the space of all matrices. Since the Gram matrix is symmetric, the basis elements are also symmetrized $v_i^T\otimes  v_j+v_j^T\otimes  v_i$. Let the eigenvalue of $v_i$ be $\lambda_i$. Then we have an eigenbasis
\begin{align}
    C(\Phi_m)^T(v_i^T\otimes  v_j+v_j^T\otimes  v_i)C(\Phi_m) = \lambda_i \lambda_j (v_i^T\otimes  v_j+v_j^T\otimes  v_i).
\end{align}
The eigenspace we seek $\lambda_i\lambda_j=1$ corresponds to $\lambda_j=\lambda_i^{-1}$.

Recall that eigenvalues of $C(\Phi_m)$ are $e^{2\pi i r/m}$ with $r$ ranging over totatives of $m$. Denote the list of the smaller half of the totatives of $m$ by $r_i$. Let the eigenvectors corresponding to $e^{2\pi i r_i/m}$ and $e^{-2\pi i r_i/m}$ be $v_i,v_{i}^*$. Then the solution space to \eqref{eq:isometry-matrix} is spanned by
\begin{align}
    \alpha_i := v_i^T \otimes v_i^{*} + v_i^{T*} \otimes v_i = 2\Re (v_i^T \otimes v_i^*).
\end{align}
Note that the basis consists of real matrices.

The index $(r;s)$ of an arbitrary solution
\begin{align}\label{eq:solution-sp}
    G=\sum_i c_i \alpha_i
\end{align}
is given by the number $r$ of negative coefficients $c_i<0$ and the number $s$ of positive coefficients $c_i>0$, since the set of $v_i$ are an eigenbasis for $G$.

Under field extensions, the dimension of the space of solutions don't change. Therefore the solution space in $\mathbb Q$ has the same dimension as that in $\mathbb R$ \eqref{eq:solution-sp}, and is spanned by some linear combinations of elements
\begin{align}
    \beta_i&=\sum_{j}c_{ij}\alpha_j,\\
    G &= \sum_i d_i \beta_i,
\end{align}
with $c_{ij}$ real, $d_i\in\mathbb Q$ and $\beta_i\in\mathbb Q^{\phi(m)\times \phi(m)}$. Since any consistent signature $(r;s)$ is possible with reals, it is also possible for rationals by using rational approximation to reals. 

Finally, for any rational solution $G=\sum_i d_i\beta_i$ to \eqref{eq:isometry-matrix} with desired signature $(r;s)$, we can obtain an integer matrix without changing the signature by multiplying by the common denominator of all entries of $G$.
\end{proof}
Putting blocks of $C(\Phi_m)$ together, we see that any $\theta\in\mathrm{GL}(n,\mathbb Z)$ can be used to construct an even lattice of any signature $(r,s)$.
\begin{corollary}\label{cor:lattice-aut-exists}
    Let $n=r+s$ and $\theta \in \mathrm{GL}(n,\mathbb Z)$ with order $m$. If $n$ satisfies \eqref{eq:rank-ineq}, then there exists an even lattice $\Lambda^{r;s}$ with isometric automorphism $\theta \in \mathrm{Aut}(\Lambda^{r;s})$.
\end{corollary}
\subsubsection{Symmetries of the embedding}\label{sub:sym-embed}
Given $\theta\in\mathrm{Aut}(\Lambda^{r;s})$, we aim to choose an embedding $\Lambda^{r;s}\hookrightarrow \Gamma^{r;s}\subset \mathbb R^{r;s}$ such that the isometric automorphism is preserved $\theta \in \mathrm{Sym}(\Gamma^{r;s})$.

To choose an embedding that manifestly respects $\theta$, choose orthonormal vectors from the rotation planes of the $\theta$, normalize, and send them to the standard basis in $\mathbb R^{r;s}$. This way, $\theta$ decomposes to left and right components manifestly as
\begin{align}
    \theta = (\theta_L;\theta_R) \in \mathrm{O}(r)\times \mathrm{O}(s).
\end{align}
\subsection{Quasicrystalline compactifications}\label{app:quasi}
Recall compactification on a $d$-dimensional torus $T^d$ is described by the Narain lattice $\Gamma^{d+x;d}$ where $x=16$ for heterotic and $x=0$ otherwise. 

Any symmetry $\theta=\mathrm{Sym}(\Gamma^{d+x;d})$ can be written by Theorem \ref{thm:decomposition} as a block matrix
\begin{align}
    \begin{pmatrix}
        C(\Phi_{d_1}) & 0 & \cdots & 0\\
        0 & C(\Phi_{d_2}) & \cdots & 0\\
        \vdots & \vdots & \ddots & \vdots\\
        0 & 0 & \cdots & C(\Phi_{d_k})
    \end{pmatrix},
\end{align}
where $C(\Phi_{d_i})$ are companion matrices to cyclotomic polynomials $\Phi_{d_i}(x)$, with eigenvalues $e^{2\pi i r/d_i}$ where $0<r<d_i$ and $\mathrm{gcd}(r,d_i)=1$. We see that if the rotation planes of any of the $C(\Phi_{d_i})$ are not all on the left or right, $\theta$ is a quasicrystalline symmetry.

We now consider quasicrystalline compactifications corresponding to a single block $C(\Phi_m)$. These are called irreducible quasicrystals. We denote such a unimodular irreducible quasicrystal as $\Gamma^{r;s}_m$. The possibilities for $m$ are given in Corollary \ref{cor:prime-pow} for $n=1$. As deduced from the corollary, all such irreducible quasicrystals exist. We give a list of them in Table \ref{tab:unimodular-quasi-body}.

Multi-block (reducible) quasicrystals can be constructed by using multiple irreducible quasicrystals, gluing various even lattices $\Lambda^{r;s}$ of Corollary \ref{cor:lattice-aut-exists}, or by explicit construction as in appendix \ref{app:construction}.

\section{Gluing construction}
\label{app:glue}
Given a lattice $\Lambda^{r;s}$ and chosen a minimal number of generators
\begin{align}
    \Lambda = \langle v_1,v_2,\dots,v_d\rangle,
\end{align}
one can define the \textit{Gram matrix} of the lattice as
\begin{align}
    G_{\Lambda^{r;s},v}:=\begin{pmatrix}
         v_1 \cdot v_1 &  v_1\cdot v_2 & \cdots &  v_1\cdot v_d\\
         v_2\cdot v_1 & \ddots & & \\
        \vdots & & &\vdots \\
         v_d\cdot v_1 & \cdots & &  v_d\cdot v_d
    \end{pmatrix}.
\end{align}
The Gram matrix is symmetric.
The Gram matrix of the dual lattice with basis $\langle w_i, v_j\rangle = \delta_{ij}$ is the inverse that of the original lattice
\begin{align}
    G_{\Lambda^{r;s*},w} = G_{\Lambda^{r;s},v}^{-1}.
\end{align}
It follows that if $\Lambda^{r;s}$ is a unimodular lattice, then the determinant of its Gram matrix is 1
\begin{align}
    \det G_{\Lambda^{r;s},v}  =1.
\end{align}
The failure of an integral lattice to be unimodular is measured by the \textit{discriminant group}
\begin{align}
    \mathcal D(\Lambda) = \Lambda^*/\Lambda.
\end{align}
It is also called the \textit{glue group} because the dual lattice can be constructed by gluing copies of $\Lambda$ along representatives of elements of $\mathcal D(\Lambda)$
\begin{align}
    \Lambda^* = \coprod_{[r_i]\in \mathcal D(\Lambda)} r_i+\Lambda
\end{align}
where $[r_i]\neq [r_j]$ for $i\neq j$, and then using the $\mathbb Q$-linear extension of the bilinear form.

The discriminant group inherits the quadratic form of the lattice as
\begin{align}
    \bar q &: \mathcal D(\Lambda) \to \mathbb Q/2\mathbb Z\\
    \bar q&([v]):=q(v)\pmod{2}.
\end{align}
If the glue group of two lattices are \textit{isometric}, i.e. there exists
\begin{align}
    \bar\psi :\mathcal D(\Lambda_1) \to \mathcal D(\Lambda_2)\\
    \bar q_1([v])=\bar q_2(\bar\psi([v])),
\end{align}
then one can glue the two lattices along their glue groups and obtain a unimodular lattice.
\begin{lemma}[Gluing Lemma]\label{lem:glue}
    If $(\Lambda_1,q_1)$ and $(\Lambda_2,q_2)$ are even lattices with isometry
    \begin{align}
        \bar \psi: \mathcal D(\Lambda_1) \to \mathcal D (\Lambda_2),
    \end{align}
    then
    \begin{align}
        \Gamma:= \{(v,w)\in \Lambda_1^*\oplus \Lambda_2^* \mid \bar\psi([v])=[w]\}
    \end{align}
    equipped with
    \begin{align}
        q(v,w):=-q_1(v)+q_2(w)
    \end{align}
    is an even unimodular lattice.
\end{lemma}
If the discriminant group is a cyclic group $\mathcal D(\Lambda)\cong \mathbb Z_N$, we can label representatives $v_i$ from each element $[v_i]$ using a label $i\in \mathbb Z_N$ by choosing the generator $[v_1]$ to be the shortest vector in $\mathcal D(\Lambda)$. This labeling can be used to write the \textit{glue code} for a gluing construction as
\begin{align}
    \Lambda_1\Lambda_2[ij]:= \coprod_{n\in \mathbb Z_N} (\Lambda_1,\Lambda_2)+n(v_i,w_j),
\end{align}
where $[v_i]\in\mathcal D(\Lambda_1)$ and $[w_j]\in\mathcal D(\Lambda_2)$.

\section{Construction of quasicrystals}\label{app:construction}
In this appendix we describe our method to explicitly construct quasicrystals of any signature and not necessarily unimodular.

Lattice automorphisms can be decomposed to blocks as in Theorem \ref{thm:decomposition}. We first show the construction of irreducible quasicrystals, which correspond to only one block.

Choose the order $m$ of the quasicrystalline action $\theta=C(\Phi_m)$. Choose the signature $(r;s)$ such that $r,s$ are even and $r+s=\phi(m)$. Choose the twist vector
\begin{align}
    (k_1,\dots,k_{r/2};k_{r/2+1},\dots,k_{(r+s)/2})/m
\end{align}
with $0<k_i<\lfloor\frac{m}{2}\rfloor$, $\mathrm{gcd}(k_i,m)=1$ and each $k_i$ distinct.

Take the starting vector $v\in \mathbb R^{r;s}$
\begin{align}\label{eq:start-v}
    v= (v_1,0,v_2,0,\dots,v_{r/2},0;v_{r/2+1},0,\dots,v_{(r+s)/2},0).
\end{align}
We will generate a set of generators by repeated application of the action in a suitable basis
\begin{align}
    B^{-1}\theta B^{-1} = \begin{pmatrix}
        R(2\pi k_1/m) & 0 & \cdots & 0\\
        0 & R(2\pi k_2/m) & \cdots & 0\\
        \vdots & \vdots & \ddots & \vdots\\
        0 & 0 & \cdots & R(2\pi k_{(r+s)/2}/m)
    \end{pmatrix}
\end{align}
where $R(\varphi)$ is the $2\times 2$ rotation matrix by $\varphi$ radians. Since $\theta$ satisfies
\begin{align}
    \theta^{\phi(m)} = -a_{\phi(m)-1}\theta^{\phi(m)-1}-\dots -a_1 \theta -a_0,
\end{align}
the generators of the lattice are
\begin{align}
    v, \theta v, \theta^2 v,\dots, \theta^{\phi(m)-1}v.
\end{align}
The lattice $\Gamma^{r;s}$ is characterized by its Gram matrix. Let
\begin{align}
    a_k = v \cdot \theta^k v.
\end{align}
Then the Gram matrix is
\begin{align}\label{eq:ak-var}
G_{\Gamma^{r;s},v}=\begin{pmatrix}
    v\cdot v & v\cdot \theta v & v\cdot \theta^2 v & \dots & v\cdot \theta^{\phi(m)-1}v\\
    v\cdot \theta v & v\cdot v & v\cdot \theta v  &  & \\
    v \cdot \theta^2 v & v\cdot \theta v & v\cdot v &\\
    \vdots &  &  &\ddots \\
    v\cdot \theta^{\phi(m)-1}v 
    & & & & v\cdot v
\end{pmatrix}=\begin{pmatrix}
    a_0 & a_1 & \cdots & a_{\phi(m)-1}\\
    a_1 & a_0 & \cdots & \vdots \\
    \vdots & & \ddots & \\
    a_{\phi(m)-1} & \cdots & & a_0
\end{pmatrix}.
\end{align}
For the lattice to be even, we need $a_0\in 2\mathbb Z$ and $a_i\in \mathbb Z$ for $i>0$.

There are $\frac{r+s}{2}$ many unknowns in $v$ \eqref{eq:start-v}, and $\phi(m)=r+s$ many $a_k$ variables \eqref{eq:ak-var}. Therefore $a_k$ are linearly dependent on each other. We choose $(r+s)/2$ many independent $a_k$ and write the other $a_i$ and $v_i$ as in terms of them.

As the last step, we try with a computer various $a_0\in 2\mathbb Z$ and for $i>0$ various $a_i\in \mathbb Z$ values until we get a solution for which all $v_i$ values are real. If there is a desired value for $\det G_{\Gamma^{r;s},v}$ as well (for unimodular lattices, we want a determinant of one for example), we keep constructing quasicrystals until we find it. 

All possibilities for $m$ listed by lattice rank are given in Table \ref{tab:all-phim}. For unimodular quasicrystals, the explicit Gram matrix entries $a_i$ are provided in Table \ref{tab:unimodular-quasi-data}. The Gram matrix entries $a_i$ of some non-unimodular quasicrystals are provided in Table \ref{tab:nonunimodular-quasi-data}.

Now we construct a reducible unimodular quasicrystal. A straightforward construction is to just take two unimodular irreducible quasicrystals together as $\Gamma_{m}^{r;s}+\Gamma_{m'}^{r';s'}$. However, it is also possible to construct unimodular quasicrystals by gluing together two nonunimodular quasicrystals.

As an example, take $\Gamma_5^{2;2}$.\footnote{Note this is the same as $\Gamma_{10}^{2;2}$ since every lattice has a symmetric $\mathbb Z_2$ action.} The basis is given by
\begin{align}
    v_1 &= \frac{\sqrt{2}}{\sqrt[4]{5}}\left(1,0;1,0\right),\\
    v_2 &= \frac{\sqrt{2}}{\sqrt[4]{5}}\left(\frac{\sqrt{5}-1}{4},\frac{\sqrt{5+\sqrt{5}}}{2\sqrt{2}};\frac{-\sqrt{5}-1}{4},\frac{\sqrt{5-\sqrt{5}}}{2\sqrt{2}}\right),\\
    v_3 &= \frac{\sqrt{2}}{\sqrt[4]{5}}\left(\frac{-\sqrt{5}-1}{4},\frac{\sqrt{5-\sqrt{5}}}{2\sqrt{2}};\frac{\sqrt{5}-1}{4},-\frac{\sqrt{5+\sqrt{5}}}{2\sqrt{2}}\right),\\
    v_4 &= \frac{\sqrt{2}}{\sqrt[4]{5}}\left(\frac{-\sqrt{5}-1}{4},-\frac{\sqrt{5-\sqrt{5}}}{2\sqrt{2}};\frac{\sqrt{5}-1}{4},\frac{\sqrt{5+\sqrt{5}}}{2\sqrt{2}}\right).
\end{align}
The determinant of the lattice is $\det G_{\Gamma_5^{2;2}}=5$ with discriminant group $\mathcal D(\Gamma_5^{2;2})=\mathbb Z_5$. 

As described in \autoref{app:glue}, we can glue two copies of $\Gamma^{2;2}_5$ together along the generator $w_1$ of its discriminant group, so
\begin{align}
    \Gamma^{2;2}_5\Gamma^{2;2}_5[11] = \coprod_{n=1}^5 (\Gamma_5^{2;2}(-1)\oplus \Gamma_5^{2;2})+n(w_1,w_1).
\end{align}
Note the first copy has its quadratic form flipped as prescribed in \autoref{lem:glue}. To obtain the generator $w_1$ of the discriminant group, one chooses the shortest vector that is in $(\Gamma_5^{2;2})^*$ but not in $\Gamma^{2;2}_5$
\begin{align}
    w_1 &= -\frac 2 5 v_1 + \frac 1 5 v_2 -\frac 1 5 v_3 +\frac 2 5 v_4,\\
    &= \frac{1}{\sqrt[4]{20}}\left(-1 , -\sqrt{1-\frac 2 {\sqrt{5}}};-1,\sqrt{1+\frac{2}{\sqrt{5}}}\right),\\
    w_1^2 &=\frac 2 5.
\end{align}
We now check if the glue vectors are preserved under the quasicrystalline action.
\begin{align}
    \theta w_1 &= -\frac 2 5 v_2 + \frac 1 5 v_3 -\frac 1 5 v_4 +\frac 2 5 (-v_1-v_2-v_3-v_4)\\
    &= -\frac 2 5 v_1 - \frac 4 5 v_2 -\frac 1 5 v_3 -\frac 3 5 v_4.
\end{align}
We see that
\begin{align}
    \theta w_1 \equiv w_1 \pmod{\Gamma^{2;2}_5}.
\end{align}
This means that the quasicrystalline action of each $\Gamma^{2;2}_5$ is a symmetry of the Narain lattice $\Gamma^{2;2}_5\Gamma^{2;2}_5[11]$.

As a result, $\Gamma_5^{2;2}\Gamma_5^{2;2}[11]$ is a reducible unimodular quasicrystalline lattice of signature $(4;4)$, whose qusaicrystalline symmetries are generated by twist vectors
\begin{align}
    \theta_1 &= (1,0;2,0)/5,\\
    \theta_2 &= (0,1;0,2)/5.
\end{align}
\begin{table}
\begin{align*}
    \begin{array}{|c|c|}
    \hline
        \text{Lattice }\Gamma^{r;s} \text{ rank } \phi(m)=r+s & \text{Symmetry order }m \\
        \hline
        4 &  5,8,10,\textbf{12}\\
        6 & 7, 9, 14, 18\\
        8 & \textbf{15}, 16, \textbf{20}, \textbf{24}, \textbf{30}\\
        10 & 11,22\\
        12 & 13, \textbf{21}, 26, \textbf{28}, \textbf{36}, \textbf{42}\\
        14 & \varnothing\\
        16 & 17, 32, 34, \textbf{40}, \textbf{48}, \textbf{60}\\
        18 & 19, 27, 38, 54\\
        20 & 25, \textbf{33}, \textbf{44}, 50, \textbf{66}\\
        22 & 23,46\\
        24 & \textbf{35}, \textbf{39}, \textbf{45}, \textbf{52}, \textbf{56}, \textbf{70}, \textbf{72}, \textbf{78}, \textbf{84}, \textbf{90}\\
        26 & \varnothing\\
        28 &  26, 58\\
        30 &  31,62\\
        32 & \textbf{51}, 64, \textbf{68}, \textbf{80}, \textbf{96}, \textbf{102}, \textbf{120}\\
        \hline
    \end{array}
\end{align*}
\caption{All solutions to $\phi(m)=r+s$. The ones with unimodular lattice realizations are in bold.}
\label{tab:all-phim}
\end{table}
\begin{table}
\begin{align}
        \begin{array}{|c|c|c|}
        \hline
            \text{Signature} & \text{Twist} & a_i \\
            \hline
            (2;2) & (1;5)/12 & (0,-1,0,0)\\
            (4;4) & (1,2;4,7)/15 & (0,-1,0,0,1,0,0,0)\\
            (4;4) & (1,3;7,9)/20 & (0,-1,0,0,0,0,0,0)\\
            (4;4) & (1,5;7,11)/24 & (0,-1,-1,0,0,-1,0,1)\\
            (4;4) & (1,7;11,13)/30 & (-2,-1,1,0,-2,-1,1,1)\\
            (6;6) & (1,4,5;2,8,10)/21 & (0,-1,0,1,1,-1,-2,0,3,1,-2,-2)\\
            (10;2) & (5,11,13,17,19;1)/42 & (-2,0,1,0,0,1,1,-1,-1,0,0,0)\\
            \hline
        \end{array}
\end{align}
\caption{The Gram matrix $G_{\Gamma^{r;s},v}$ entries for irreducible unimodular quasicrystals. The $ij$-th entry of the Gram matrix is $a_{|i-j|}$.}
\label{tab:unimodular-quasi-data}
\end{table}
\begin{table}
\begin{align}
        \begin{array}{|c|c|c|c|c|}
        \hline
            \text{Signature} & \text{Twist} & a_i & \mathcal{D}(\Lambda) \\
            \hline
            (2;2) & (1;3)/8 & (0,-1,0,1) & \mathbb Z_2^2\\
            (2;2) & (1;3)/10 & (0,-1,-1,1) & \mathbb Z_5\\
            (4;2) & (2,3;1)/7 & (0,1,0,-1,-1,0) & \mathbb Z_7\\
            (6;2) & (3,7,9;1)/20 & (0,1,1,1,1,0,-1,-1) & \mathbb Z_2^4\\
            (6;2) & (5,7,11;1)/24 & (0,1,1,1,0,0,0,0) & \mathbb Z_2^2\\
            (6;2) & (7,11,13;1)/30 & (-2,1,1,0,0,-1,1,1) & \mathbb Z_5^2\\
            \hline
        \end{array}
\end{align}
\caption{Explicit construction data for non-unimodular quasicrystals. As before, $a_{|i-j|}$ corresponds to the $ij$-th entry in the Gram matrix $G_{\Lambda^{r;s},v}$. The last column gives the discriminant group $\mathcal D(\Lambda)$ of the lattice. Since the lattices are non-unimodular, the discriminant group is nontrivial.}
\label{tab:nonunimodular-quasi-data}
\end{table}

\section{Freely acting orbifolds}\label{app:free}
The orbifold obtained by a group $G$ acting without fixed points is termed a \textit{freely acting orbifold}. The most useful feature of these orbifolds is that their twisted sectors can be massed up.

A cyclic orbifold action always has even number of eigenvalues. Therefore, in odd dimensions, there is always at least one real dimension that is fixed. We can shift along the fixed real dimension together with the rotations to get a freely acting orbifold.

More explicitly, the technique we use in odd dimensions involves orbifolding on $T^d$ coupled with a shift on an additional $S^1$, compactifying overall to $9-d$ dimensions. The advantage of this construction is that except at special $S^1$ radii, all twisted sectors become massive. Intuitively, as the radius $r$ increases, twisted sector strings become longer and thus gain mass. 

Since the $S^1$ remains invariant under the overall orbifolding action, we have
\begin{align}
    I \supset \Gamma^{1,1} = \left\{\frac 1 2(n/r + w r,n/r - w r)|n,w\in\mathbb Z\right\}\,.
\end{align}
One can then choose a shift vector that satisfies level matching and also $r$ large enough to lift the twisted sectors. 

Essentially, freely acting orbifolds enable us to project out a significant portion of the massless spectrum in the untwisted sector without introducing massless states in the twisted sectors. For more details, we refer to appendix A.3 of \cite{Baykara:2023plc}.

\section{Quantum symmetry of orbifolds}\label{app:quantsym}
An interesting effect of orbifolding is that 
 you have a CFT $\mathcal C$ and orbifold by an abelian $g$, the resulting CFT $\mathcal C'$ has the same symmetry $g$ now called ``quantum symmetry''\cite{Vafa:1989ih}. In fact, gauging this symmetry again corresponds to ungauging and one ends up with the original CFT $\mathcal C$.

In particular, consider $Z_{T^n}[M,A]$  the partition function of $M$ on $T^n$ with some background gauge field $A$ of group $G$ with $A$ taking values in $H^1(M,G)$. Equivalent for the orbifolded theory on $T^n/G$ consider the partition function $Z_{T^n/G}[M,A']$ with $A'$ being the gauge field of some group $G'$ and taking values in $H^1(M,G')$ then we know that 
\begin{eqnarray}
   Z_{T^n/G}[M,A']\propto \sum_A e^{i(A',A)}Z_{T^n}[M,A]
\end{eqnarray}
with 
\begin{eqnarray}
    e^{i(\underline{\ },\underline{\ })}:H^1(M,G')\times H^1(M,G)\to H^2(M,U(1))\equiv U(1)
\end{eqnarray}
we can invert this expression by thinking of $M$ on $T/G/G'$ as 
\begin{eqnarray}
Z_{T^n}[M,A]   \propto \sum_A' e^{i(A',A)}Z_{T^n/G}[M,A']
\end{eqnarray}
being a version of a discrete Fourier transform.

 \section{Normalization of U(1)s }\label{app:abelian}
The OPE of a current algebra with generators $J^a$ is given by 
 \begin{eqnarray}\label{current}
     J^a({z})J^b({w})
\sim {k\delta^{ab}\over ({z}-{w})^2}+{if^{abc}J^c({w})\over ({z}-{w})} \end{eqnarray}
In the heterotic string we have both levels equal to $k=1$ such that the central charge of the Kac-Moody alegbra is $c=2k{248\over k+30}=16$ as expected.

In the Cartan basis we can express them as 
\begin{align}
E_a=e^{iP_a\cdot X}, \ p^2_a=2 \\
H_i=i h_i\cdot \partial X
\end{align}
Then the level of the current algebra \autoref{current} can be expressed as $k\to \hat{k}={2k\over \psi^2}$ where $\psi^2$ is the length of the highest root which for the heterotic string is $\psi^2=2$ and hence $\hat{k}=k$. Therefore, when we consider compactifications that break the heterotic gauge groups we still want to normalize the level appropriately.

The U(1) factors that we consider in this work come from the untwisted sector and hence from a Cartan element of $E_8$ therefore it can be written as 
\begin{eqnarray}
    J_0=V_Q\cdot \partial X
\end{eqnarray}
The level of this U(1) is given by 
\begin{eqnarray}
  k_{U(1)}=|V_Q|_{\alpha}^2
\end{eqnarray}
In this work we work in the $\alpha-$basis and hence the norm is taken with respect to the quadratic form of the lattice.
A state that is charged under the gauge group will carry a momentum vector $P$ that is in the charge lattice in the untwisted sector or in the shifted lattice of the twisted sectors and hence have a vertex operator of the form $e^{iP\cdot  X}$ and hence for the U(1) generator $J_0=V_Q\cdot \partial X$ the OPE will be
\begin{equation}
    J_0(z)e^{-P\cdot X(z)}\sim {V_Q\cdot P\over  (z-w)}+\cdots 
\end{equation}
In other words the U(1) charge is given by $Q=V_Q\cdot P$ which specifies also the conformal dimension of the state as  $\Delta={Q^2\over 2 k_{U(1)}}$. If we want the level of the abelian algebra to be 1 then we need to normalize the charges accordingly by dividing with the norm of the basis vector as $Q\to {Q\over {|V_Q|_\alpha}}$.

 \section{Details on 6d constructions}

A sort review of the 6d minimal supergravity and chiral anomalies can be found in Appendix B of  \cite{Baykara:2023plc}. Interestingly in the case of a single tensor multiplet the anomaly polynomial factorizes. In this work the 6d theories we study mainly come from heterotic models and hence only one tensor multiplet will be present.

 \begin{eqnarray}
 I_8(R,F)=\frac{1}{2}\Omega_{\alpha\beta}X^\alpha_4 X^\beta_4, \ \ X_4^\alpha=\frac{1}{2}a^\alpha trR^2 +\sum_ib_i^\alpha \frac{2}{\lambda_i}trF_i^2 +2b_{ij}^\alpha F_iF_j
 \end{eqnarray}
 where $a^\alpha, b_i^\alpha$ are vectors in $\R^{1,T}$, $\Omega_{\alpha \beta }$ is the metric on this space and  $ \lambda_i $ are normalization factors of the gauge groups $G_i$.

 For $T=1$ the decomposition is given by 

 \begin{eqnarray}
   \hspace{-1cm}
 I_8={1\over 16}(trR^2-\sum_k\alpha_k trF_k^2-\sum_{ij}\alpha_{ij}F_iF_j)\wedge (trR^2-\sum_k\tilde{\alpha}_k trF_k^2-\sum_{ij}\tilde{\alpha}_{ij}F_iF_j)
 \end{eqnarray}
in the basis that 
\begin{eqnarray}
  \hspace{-1cm}
    \Omega=\begin{pmatrix}
        0 & 1 \\ 1 & 0
    \end{pmatrix}, a=\begin{pmatrix}
        -2 \\-2
    \end{pmatrix}, b_k={1\over 2 }\lambda_k 
\begin{pmatrix}
    \alpha_k \\ \tilde{\alpha}_k
\end{pmatrix},\ b_{ij}={1\over 2 }\lambda_{ij} 
\begin{pmatrix}
    \alpha_{ij} \\ \tilde{\alpha}_{ij}
\end{pmatrix}, \ j={1\over \sqrt{2}}\begin{pmatrix}
    e^\phi \\ e^{-\phi}
\end{pmatrix}
\end{eqnarray}

Let us consider all the anomaly conditions:

 \begin{itemize}
 	\item ${\label{eqn:R4} R^4:  \ \ H-V=273-29T}$
 	\item$ {\label{eqn:F4}trF^4: \ \ 0=B^i_{Adj}-\sum n_R^i B^i_R}$
 	\item ${\label{eqn:R22}(R^2)^2: \  a\cdot a=a^\alpha\Omega_{\alpha \beta }a^\beta  =9-T}$
 	\item ${\label{eqn:F2R2} trF^2R^2: \ -a\cdot b_i=-a^\alpha\Omega_{\alpha \beta }b_i^\beta   =-\frac{1}{6}\lambda_i (A^i_{Adj}-\sum_Rn_R^iA^i_R)} = \alpha_i +\tilde{\alpha_i}$
  \item ${\label{eqn:aF2R2} F_i F_jR^2: \ -a\cdot b_{ij}=\frac{1}{6}\sum_I M_I q_{I,i}q_{I,j}} = \alpha_{ij} +\tilde{\alpha_{ij}}$
 	\item ${\label{eqn:F22}(trF^2)^2: \  b_i \cdot b_i =b_i^\alpha\Omega_{\alpha \beta }b_i^\beta  =\frac{1}{3}\lambda_i^2 (\sum_R n_R^iC^i_R-C^i_{Adj}) }={1\over 2}\alpha_i \tilde{\alpha_i}$
  \item ${\label{eqn:aF22}F_iF_jF_kF_l: \  b_{ij}\cdot b_{kl}+b_{ik}\cdot b_{jl}+b_{il}\cdot b_{kj}  =\sum_I M_I q_{I,i}q_{I,j}q_{I,k}q_{I,l} }$ $$\quad \quad \quad \quad \quad \quad \quad \quad \quad \quad \quad \quad ={1\over 2}(\alpha_{ij} \tilde{\alpha_{ik}}+\alpha_{ik} \tilde{\alpha_{jl}}+\alpha_{il} \tilde{\alpha_{kj}})$$
  \item ${\label{eqn:aFF}F_iF_jtrF_k^2: \  b_k \cdot b_{ij} =\sum_I  M_I^k \lambda_{k} A_Rq_{I,i}q_{I,j}  }$ 
  \item $F_i^3trF_k: 0= \sum_I M^k_IE^I_kq_{I,i}$
 	\item ${\label{eqn:F2F2}trF^2_itrF^2_j: \ b_i \cdot b_j=b_i^\alpha\Omega_{\alpha \beta }b_j^\beta  = \sum_{R,S}\lambda_i \lambda_jn_{RS}^{ij}A^i_RA^j_S ={1\over 4}(\alpha_i \tilde{\alpha_j}+\alpha_j \tilde{\alpha_i}) \ \ \ i\neq j }$
 	 \end{itemize}
where $tr_RF^3=E_RtrF^3$.

\subsection*{Theory 1}
\begin{equation}
E_8\times E_7\times U(1)
\end{equation}

\begin{itemize}
    \item \textbf{Untwisted Sector}
    
    1 \text{Tensor}
    \item \textbf{Twisted Sector}
    \begin{align}
(\textbf{1},\textbf{56})_{(5,4,3,2,1,0)}^{(1,1,2,3,1,2)} +
(\textbf{1},\textbf{1})_{(5,4,3,2,1,0)}^{(1,1,2,3,1,2)}    
\end{align}
\end{itemize}

The total for $E_8$ is $10 (\textbf{56})$ and $66$ charged under U(1) only giving a total of $H_c=626$ which satisfies the gravitational anomalies.

We have that 
\begin{equation}
-6 a\cdot b_{U(1)}=    \sum_i n_i q_i^2=42, \  3 b_{U(1)}^2=\sum_i n_i q_i^4=9, b_{U(1)}\cdot b_{E_7}= \sum_i {12 }n_{56,i} q_i^2=6
\end{equation}
Note that in fact for $T=1$ these sums are fixed by anomalies and hence matching our expectations.

\begin{equation}
    E7: A_{56}=1,C_{56}=1/24,A_{adj}=3,C_{adj}=1/6
\end{equation}
\begin{equation}
    E8: A_{adj}=1,C_{adj}=1/100
\end{equation}

with anomaly lattice:
\begin{equation}
    \begin{pmatrix}
        8 &-14& 10& -14 \\
        -14 & 12& 0 & 12\\
        10 & 0 &-12 &0 \\
        -14&12 & 0 & 12
    \end{pmatrix}
\end{equation}
\begin{equation}
a=(-2,-2),   \  b_{E_7}=(1,6)=b_{U(1)},  \  b_{E_8}=(1,-6)
\end{equation}
such the anomaly polynomial factorizes as
\begin{equation}
    {-1\over 16 }(trR^2-{1\over 6}trF_{E_7^2}-{1\over 30}tr^2_{E_8}-2tr^2_{U(1)})\wedge (trR^2-trF_{E_7^2}+{1\over 5}tr^2_{E_8}-12tr^2_{U(1)})
\end{equation}
The expectation is that this theory corresponds  to an elliptic threefold with the base $\mathbb{F}_{12}$ since it contains the non-Higgsable $E_8$ and $10\times \textbf{56}$ can completely Higgs the $E_7$.

Note that the references \cite{Park:2011wv,Erler:1993zy} are off by an order of 2 for the values of the abelian vectors of this model. Their error traces back to the correct normalization of the currents as described in \autoref{app:abelian}. In our case the  U(1) basis is $V^1_Q=(6,0^{16})$ with $(V^1_Q)^2=72$ (in $\alpha$-basis). Therefore as explained in \autoref{app:abelian} the physical normalization for the charges it $Q'={Q\over \sqrt{72}}$.

Note that as one would expect the untwisted sector has no other moduli other than the dilaton. This indicates that as described in \autoref{sec2:lattices} the quasicrystalline action fixes the $G_{ij},B_{ij}$ Narain moduli.

\subsection*{Theory 2}

\begin{equation}
 E_8 \times SO(12) \times SU(2)      \times U(1)   
\end{equation}

\begin{itemize}
    \item \textbf{Untwisted Sector}
\begin{equation*}
(\textbf{1},\bar{\textbf{32}},\textbf{1})_{(1,2,1)}  ^{(6,2,0)} 
\end{equation*}
\item \textbf{Twisted Sectors}
\begin{align*}  &(\textbf{1},\bar{\textbf{32}},\textbf{1})_{(2,1)}  ^{(2,0)}   +(\textbf{1},\textbf{32},\textbf{1})_{(5,3,1)}^{(1,2,1)}  +(\textbf{1},{\textbf{12}},\textbf{2})_{(4,2,0)}^{(1,3,2)}  \\
&  +(\textbf{1},{\textbf{1}},\textbf{2})_{(7,5,9,3)}^{(5,7,1,3)} +(\textbf{1},\textbf{1},\textbf{1})_{(2,10,4,8,6)}^{(2,2,12,6,12)} 
\end{align*}
\end{itemize}

All anomaly cancellation conditions are satisfied and the  anomaly coefficients are given by 
\begin{align}
   & b_{SO(12)}^2=12,     a\cdot b_{SO(12)}=-14\\
     & b_{SU(2)}^2=12,     a\cdot b_{SU(2)}=-14\\
      & b_{U(1)}^2=12,     a\cdot b_{U(1)}=-14\\
    & b_{E_8}^2=-12,     a\cdot b_{E_8}=10\\
\end{align}
where solutions are given by 
\begin{equation}
a=(-2,-2),   \  b_{SO(12)}=(1,6)=b_{SU(2)}=b_{U(1)},  \  b_{E_8}=(1,-6)
\end{equation}
The charges are expressed in the U(1) basis $V^1_Q=(0^{15},6)$ with $(V^1_Q)^2=72$ (in $\alpha$-basis). Therefore as explained in \autoref{app:abelian} the physical normalization for the charges it $Q'={Q\over \sqrt{72}}$.

\subsection*{Theory 3}

\begin{eqnarray}
    E_8\times SO(10)\times SU(3)\times U(1)
\end{eqnarray}

\begin{itemize}
    \item \textbf{Untwisted Sector} 
    \begin{eqnarray}
(\textbf{1},\textbf{16},\textbf{1})_{(15)}^{(1)} +(\textbf{1},\textbf{10},\textbf{3})_{(-10)}^{(1)} 
    \end{eqnarray}
    \item \textbf{Twisted Sectors}
    \begin{eqnarray}
(\textbf{1},\textbf{16},\textbf{1})_{(3,-9)}^{(10,5)}        +(\textbf{1},\textbf{10},\textbf{3})_{(2)}^{(5)}     +(\textbf{1},\textbf{1},\textbf{1})_{(12)}^{20)}  
    \end{eqnarray}
\end{itemize}

\begin{align}
   & b_{E_8}^2=-12,   \quad  a\cdot b_{E_8}=-10\\
    &  b_{SO(10)}^2=b_{SU(3)}^2=b_{U(1)}^2=12, \quad  a\cdot b_{SO(10)}=a\cdot b_{U(1)}=14\\
\end{align}
where solutions are given by 
\begin{equation}
a=(-2,-2),   \  b_{SO(12)}=(1,6)=b_{SU(2)}=b_{U(1)},  \  b_{E_8}=(1,-6)
\end{equation}
The charges are expressed in the U(1) basis $V^1_Q=(0^{12},10, 15, 10, -5)$ with $(V^1_Q)^2=300$ (in $\alpha$-basis). Therefore as explained in \autoref{app:abelian} the physical normalization for the charges it $Q'={Q\over \sqrt{300}}$.
\subsection*{Theory 4}
\begin{eqnarray}
   E_8\times  SU(4)\times SU(4)\times  SU(2)\times U(1)
\end{eqnarray}

\begin{align}
   & b_{E_8}^2=-12,   \quad  a\cdot b_{E_8}=-10\\
    &  b_{SO(10)}^2=b_{SU(3)}^2=b_{U(1)}^2=12, \quad  a\cdot b_{SO(10)}=a\cdot b_{U(1)}=14\\
\end{align}
\begin{equation}
a=(-2,-2),   \  b_{SO(10)}=(1,6)=b_{SU(3)}=b_{U(1)},  \  b_{E_8}=(1,-6)
\end{equation}
The charges are expressed in the U(1) basis $V^1_Q=(0^{11},4,0^4)$ with $(V^1_Q)^2=32$ (in $\alpha$-basis). Therefore as explained in \autoref{app:abelian} the physical normalization for the charges it $Q'={Q\over \sqrt{32}}$.
\subsection*{Theory 5}
\begin{eqnarray}
   SU(9)\times  SO(12)\times SU(2)\times U(1)
\end{eqnarray}
\begin{itemize}
    \item \textbf{Untwisted Sector:}
    0
    \item \textbf{Twisted Sector:}
 \begin{align}
    &  ( \textbf{9,1,2})_{(1,1)}^{(1,3)}+   ( \textbf{9,1,1})_{(2,0)}^{(2,6)}+   ( \textbf{1,1,2})_{(3,1)}^{(1,6)}+   ( \textbf{1,32,1})_{(1,0)}^{(1,1)}+   ( \textbf{1,12,2})_{(0)}^{(2)}\\ & +   ( \textbf{1,12,1})_{(1)}^{(4)}+   ( \textbf{1,1,2})_{(1)}^{(3)}+   ( \textbf{36,1,1})_{(0)}^{(2)} +  ( \textbf{1,1,1})_{(2)}^{(10)}
       \end{align} 
\end{itemize}
\begin{align}
   & b_{SU(9)}^2=b_{SO(12)}^2=0,   \quad  a\cdot b_{SU(9)}=  a\cdot b_{SO(12)}=-2\\
    &  b_{SU(2)}^2=8, \quad  a\cdot b_{SO(10)}=-10,\quad b^2_{U(1)}=4, a\cdot b_{U(1)}=-6\\
\end{align}
\begin{equation}
a=(-2,-2),   \  b_{SU(9)}=(0,1)=b_{SO(12)},  \  b_{SU(2)}=(4,1),\ b_{U(1)}=(2,1)
\end{equation}
The charges are expressed in the U(1) basis $V^1_Q=(0^{9},2, 4, 4, 4, 4, 2, 2)$ with $(V^1_Q)^2=8$ (in $\alpha$-basis). Therefore as explained in \autoref{app:abelian} the physical normalization for the charges it $Q'={Q\over \sqrt{8}}$.

\subsection*{Theory 6}
\begin{eqnarray}
   E_7\times   U(1)
\end{eqnarray}
\begin{itemize}
    \item \textbf{Untwisted Sector:}
    0
    \item \textbf{Twisted Sector:}
 \begin{align}
    &  ( \textbf{1})_{(2,\ 3,\ 4,\ 5,\ 6,\ 7)}^{(82,64,42,28,12,4)}
       \end{align} 
\end{itemize}
\begin{align}
   & b_{E_7}^2=-8,   \quad  a\cdot b_{E_7}=6\\
    &   b^2_{U(1)}=8, a\cdot b_{U(1)}=-10\\
\end{align}
\begin{equation}
a=(-2,-2),   \  b_{E_7}=(0,1),  \  b_{U(1)}=(4,1)
\end{equation}
The charges are expressed in the U(1) basis $V^1_Q=(0^{5},5,0^2)$ with $(V^1_Q)^2=50$ (in $\alpha$-basis). Therefore as explained in \autoref{app:abelian} the physical normalization for the charges it $Q'={Q\over \sqrt{50}}$.

\subsection*{Theory 7}
\begin{eqnarray}
   E_6\times   U(1)^2
\end{eqnarray}

\begin{itemize}
    \item \textbf{Untwisted Sector:}
    0
    \item \textbf{Twisted Sector:}
 \begin{align}
    &   ( \textbf{1})_{([7,0],\ [3,0],\ [4,0],\ [6,0],[0,3],[1,0],[2,0],[5,0])}^{(\ 1,  \quad \ 9,\quad  \ 14, \quad  6, \ \ 20, \  \  67, \ \ 72, \ \  27 \ )}                    \\  & ( \textbf{1})_{( [2,-3],\ [-3,3],[5,3],[5,-3],[2,3],[3,3],[1,-3],[-4,3],[6,3])}^{(\ \   6, \quad \ \ 5, \quad \ 3, \quad \ 4, \quad 24, \ \ 18,  \quad 34, \quad 12, \quad 2 \ )}       
       \end{align} 
\end{itemize}
\begin{align}
   & b_{E_6}^2=-6,   \quad  a\cdot b_{E_6}=4\\
    &   b^2_{U(1)_1}=   b^2_{U(1)_2}=6, a\cdot b_{U(1)_1}=a\cdot b_{U(1)_2}=-8\\ & b_{U(1)_1}\cdot b_{U(1)_2}=6
\end{align}
\begin{equation}
a=(-2,-2),   \  b_{E_6}=(-3,1),  \  b_{U(1)}=(3,1)
\end{equation}
The charges are expressed in the U(1) basis $V^1_Q=(0^{6},5,0),V^2_Q=(0^{5},4,2,0)$ with $(V^1_Q)^2=50,(V^1_Q)^2=24$ (in $\alpha$-basis). Therefore as explained in \autoref{app:abelian} the physical normalization for the charges it $Q'_1={Q_1\over \sqrt{50}}, \ Q'_2={Q_2\over \sqrt{24}}$.

\section{4d matter}\label{app:4dmatter}
Here, we present the representations of the charged fermions for the models in \autoref{tab:4d-N=1}. In \autoref{tab:4d-model3}, we give the charges of the chiral matter for model 3, and in \autoref{tab:4d-model4} for model 4. We only denote the left-handed Weyl fermion charges.
\renewcommand{\baselinestretch}{1.5}
\begin{table}[h!] 
\begin{align*}
\hspace{-0.5cm}
\resizebox{16cm}{!}{$ 
    \begin{array}{|ccc|}
    \hline
    \multicolumn{3}{|c|}{\text{4d }\mathcal N=1\text{ Model No. 3 Matter Charges}}\\
    \hline
 7\times (\mathbf{1},\mathbf{1},\mathbf{1},\mathbf{1},\mathbf{1})_{2, -2, -2, 2, 0, 2, 0} &1\times (\mathbf{1},\mathbf{\overline{4}},\mathbf{2},\mathbf{1},\mathbf{1})_{-2, 1, 1, -1, 2, -1, 2}&1\times (\mathbf{1},\mathbf{4},\mathbf{1},\mathbf{1},\mathbf{1})_{0, 2, 1, 0, 0, 1, 0}\\1\times (\mathbf{1},\mathbf{6},\mathbf{1},\mathbf{1},\mathbf{2})_{2, 0, 1, 0, -2, 0, -2}&1\times (\mathbf{1},\mathbf{1},\mathbf{1},\mathbf{1},\mathbf{2})_{-2, -2, -5, 2, 2, 2, 2}&1\times (\mathbf{\overline{4}},\mathbf{1},\mathbf{1},\mathbf{1},\mathbf{1})_{-1, 2, 0, 0, 0, 1, -1}\\1\times (\mathbf{1},\mathbf{1},\mathbf{1},\mathbf{2},\mathbf{1})_{2, -3, 0, 1, -1, -3, 1}&1\times (\mathbf{1},\mathbf{1},\mathbf{1},\mathbf{2},\mathbf{1})_{0, 1, 0, 1, 1, 1, 3}& 1\times (\mathbf{1},\mathbf{1},\mathbf{1},\mathbf{1},\mathbf{2})_{0, -3, -2, 1, 2, 0, 2}\\1\times (\mathbf{\overline{4}},\mathbf{1},\mathbf{1},\mathbf{1},\mathbf{1})_{1, -2, 0, 0, 2, 1, 1}& 2\times (\mathbf{1},\mathbf{4},\mathbf{1},\mathbf{1},\mathbf{1})_{2, -2, 1, 0, 0, -1, 0} &1\times (\mathbf{4},\mathbf{1},\mathbf{1},\mathbf{2},\mathbf{1})_{-1, 1, 0, -1, -1, 0, -2}\\ 2\times (\mathbf{1},\mathbf{1},\mathbf{1},\mathbf{1},\mathbf{2})_{-2, 1, -2, 1, 2, 2, 2}&2\times (\mathbf{1},\mathbf{4},\mathbf{1},\mathbf{1},\mathbf{2})_{1, -1, -2, 1, -1, 1, -1}&2\times (\mathbf{1},\mathbf{1},\mathbf{2},\mathbf{1},\mathbf{1})_{1, 2, 4, -2, -1, -2, -1}\\2\times (\mathbf{1},\mathbf{4},\mathbf{1},\mathbf{1},\mathbf{1})_{-3, 1, 1, -1, 3, -1, 3}&2\times (\mathbf{1},\mathbf{1},\mathbf{1},\mathbf{1},\mathbf{1})_{-1, 1, 0, -1, 3, 4, -1}&2\times (\mathbf{4},\mathbf{1},\mathbf{1},\mathbf{2},\mathbf{1})_{0, 0, 0, 0, 0, 0, -1}\\2\times (\mathbf{1},\mathbf{1},\mathbf{1},\mathbf{2},\mathbf{1})_{-1, 2, 0, 0, -2, -1, 0}&2\times (\mathbf{\overline{4}},\mathbf{1},\mathbf{1},\mathbf{1},\mathbf{1})_{0, -1, 0, -1, -1, -1, -2}&2\times (\mathbf{6},\mathbf{1},\mathbf{1},\mathbf{1},\mathbf{1})_{1, -1, 0, 1, 1, 0, 3}\\2\times (\mathbf{1},\mathbf{1},\mathbf{2},\mathbf{1},\mathbf{2})_{1, 0, 1, 0, -1, 0, -1}&1\times (\mathbf{1},\mathbf{4},\mathbf{1},\mathbf{1},\mathbf{1})_{-1, 0, 0, -1, 1, 0, -1}&1\times (\mathbf{1},\mathbf{1},\mathbf{1},\mathbf{2},\mathbf{1})_{-1, -2, -2, 1, 1, -1, 3}\\1\times (\mathbf{1},\mathbf{1},\mathbf{2},\mathbf{1},\mathbf{2})_{1, 0, -1, 1, 0, 2, 0}&1\times (\mathbf{\overline{4}},\mathbf{1},\mathbf{1},\mathbf{1},\mathbf{1})_{0, 2, 3, -1, -1, -2, 0}&4\times (\mathbf{1},\mathbf{1},\mathbf{1},\mathbf{1},\mathbf{2})_{1, -1, -1, 0, -1, 1, -3}\\4\times (\mathbf{1},\mathbf{1},\mathbf{1},\mathbf{1},\mathbf{1})_{1, 0, 3, -1, -1, -3, 1}&2\times (\mathbf{1},\mathbf{1},\mathbf{1},\mathbf{2},\mathbf{1})_{-2, -1, -2, 0, 2, 1, 0}&2\times (\mathbf{1},\mathbf{1},\mathbf{1},\mathbf{1},\mathbf{2})_{2, -2, -1, 1, -2, -1, 0}\\2\times (\mathbf{1},\mathbf{1},\mathbf{2},\mathbf{1},\mathbf{1})_{0, 2, 3, -1, 1, 0, 1}&6\times (\mathbf{1},\mathbf{1},\mathbf{1},\mathbf{1},\mathbf{1})_{0, 1, 3, -2, 0, -1, -2}&5\times (\mathbf{1},\mathbf{1},\mathbf{1},\mathbf{1},\mathbf{1})_{-1, 0, -1, 1, 1, -1, 3}\\5\times (\mathbf{1},\mathbf{1},\mathbf{1},\mathbf{1},\mathbf{1})_{-3, 2, -1, -1, 1, 1, -1}&5\times (\mathbf{1},\mathbf{1},\mathbf{2},\mathbf{1},\mathbf{1})_{3, -1, 1, 0, -2, 0, -2}&5\times (\mathbf{1},\mathbf{1},\mathbf{1},\mathbf{1},\mathbf{1})_{1, 2, 1, -1, -3, 1, -5}\\6\times (\mathbf{1},\mathbf{1},\mathbf{2},\mathbf{1},\mathbf{1})_{-1, -1, -1, 0, 2, 0, 2}&6\times (\mathbf{1},\mathbf{1},\mathbf{1},\mathbf{1},\mathbf{1})_{3, 0, 1, 1, -3, -1, -1}&1\times (\mathbf{1},\mathbf{1},\mathbf{1},\mathbf{1},\mathbf{2})_{1, -1, 1, 0, -1, -3, 1}\\1\times (\mathbf{1},\mathbf{1},\mathbf{1},\mathbf{1},\mathbf{1})_{-1, -2, -3, 1, 3, 3, 1}&2\times (\mathbf{1},\mathbf{1},\mathbf{2},\mathbf{1},\mathbf{1})_{2, -1, 0, 1, -1, -1, 1}&2\times (\mathbf{1},\mathbf{1},\mathbf{1},\mathbf{1},\mathbf{1})_{0, 2, 0, 0, -2, 0, -2}\\2\times (\mathbf{1},\mathbf{1},\mathbf{2},\mathbf{1},\mathbf{1})_{0, 1, 0, -1, -1, 1, -3}&2\times (\mathbf{1},\mathbf{1},\mathbf{1},\mathbf{1},\mathbf{1})_{2, -2, 0, 0, 0, 0, 0}&2\times (\mathbf{1},\mathbf{1},\mathbf{1},\mathbf{1},\mathbf{1})_{0, -2, 0, 0, 0, -2, 0}\\2\times (\mathbf{1},\mathbf{1},\mathbf{1},\mathbf{1},\mathbf{1})_{-2, 2, 0, 0, 2, 2, 2}&2\times (\mathbf{\overline{4}},\mathbf{1},\mathbf{1},\mathbf{1},\mathbf{1})_{-1, 0, 0, 0, 2, 1, 1}&2\times (\mathbf{4},\mathbf{1},\mathbf{1},\mathbf{1},\mathbf{1})_{-1, 0, 0, 0, 0, -1, 1}\\3\times (\mathbf{1},\mathbf{1},\mathbf{1},\mathbf{1},\mathbf{1})_{2, -2, 0, 0, -2, -2, -2}&3\times (\mathbf{1},\mathbf{1},\mathbf{1},\mathbf{1},\mathbf{1})_{0, 2, 0, 0, 0, 2, 0}&3\times (\mathbf{4},\mathbf{1},\mathbf{1},\mathbf{1},\mathbf{1})_{1, 0, 0, 0, -2, -1, -1}\\3\times (\mathbf{1},\mathbf{1},\mathbf{1},\mathbf{1},\mathbf{1})_{-2, 2, 0, 0, 0, 0, 0}&3\times (\mathbf{1},\mathbf{1},\mathbf{2},\mathbf{1},\mathbf{1})_{-2, 1, 0, -1, 1, 1, -1}&3\times (\mathbf{1},\mathbf{1},\mathbf{1},\mathbf{1},\mathbf{1})_{0, -2, 0, 0, 2, 0, 2}\\4\times (\mathbf{\overline{4}},\mathbf{1},\mathbf{1},\mathbf{1},\mathbf{1})_{1, 0, 0, 0, 0, 1, -1}&4\times (\mathbf{1},\mathbf{1},\mathbf{2},\mathbf{1},\mathbf{1})_{0, -1, 0, 1, 1, -1, 3}&8\times (\mathbf{1},\mathbf{1},\mathbf{1},\mathbf{2},\mathbf{1})_{0, 0, 1, 0, 2, 0, 2}\\6\times (\mathbf{1},\mathbf{1},\mathbf{1},\mathbf{2},\mathbf{1})_{-1, 1, 1, -1, -1, -2, -1}&6\times (\mathbf{1},\mathbf{1},\mathbf{1},\mathbf{2},\mathbf{1})_{-1, 1, 1, -1, 1, 0, 1}&6\times (\mathbf{1},\mathbf{1},\mathbf{1},\mathbf{1},\mathbf{1})_{1, -1, -2, 1, -3, 0, -3}\\6\times (\mathbf{1},\mathbf{1},\mathbf{1},\mathbf{1},\mathbf{1})_{1, -1, -2, 1, -1, 2, -1}& &\\
\hline
    \end{array}$ }
\end{align*}
\caption{Charged matter for model 3.}
\label{tab:4d-model3}
\end{table}
\renewcommand{\baselinestretch}{1.5}
\begin{table}[h!] 
\begin{align*}
\hspace{-0.5cm}
\resizebox{16cm}{!}{$ 
    \begin{array}{|ccc|}
    \hline
    \multicolumn{3}{|c|}{\text{4d }\mathcal N=1\text{ Model No. 4 Matter Charges}}\\
    \hline
1\times (\mathbf{1},\mathbf{1},\mathbf{2},\mathbf{1},\mathbf{1})_{0, 0, 1, 0, 1, 0, 1, 1, 0, 0, 0}&1\times (\mathbf{1},\mathbf{1},\mathbf{2},\mathbf{1},\mathbf{1})_{0, 0, -1, -1, 1, 2, -4, -1, -1, -1, 0}&1\times (\mathbf{1},\mathbf{1},\mathbf{1},\mathbf{2},\mathbf{1})_{-1, 1, 0, -1, 0, -1, 0, 1, 1, 0, 0}\\1\times (\mathbf{1},\mathbf{1},\mathbf{1},\mathbf{1},\mathbf{1})_{0, 1, -2, -2, 0, 0, -1, 0, -1, 0, 0} &1\times (\mathbf{1},\mathbf{1},\mathbf{1},\mathbf{2},\mathbf{2})_{-1, 1, 0, -2, 1, -2, 0, 1, 0, 1, 0}&1\times (\mathbf{1},\mathbf{1},\mathbf{1},\mathbf{2},\mathbf{1})_{2, -2, 0, 2, -1, 4, 0, -2, 0, -1, 0}\\1\times (\mathbf{1},\mathbf{1},\mathbf{1},\mathbf{2},\mathbf{1})_{0, 0, 0, 1, 0, 2, 0, -1, 0, 1, 0}&1\times (\mathbf{1},\mathbf{1},\mathbf{1},\mathbf{2},\mathbf{2})_{2, 0, 0, 0, -1, 2, 0, -1, 0, 0, 0}&1\times (\mathbf{1},\mathbf{1},\mathbf{1},\mathbf{1},\mathbf{2})_{-2, 0, 0, -1, 1, -4, 0, 2, 0, -1, 0}\\1\times (\mathbf{1},\mathbf{1},\mathbf{1},\mathbf{2},\mathbf{1})_{1, -2, 0, 1, -1, 1, 0, 0, 1, -3, 0}&1\times (\mathbf{1},\mathbf{1},\mathbf{1},\mathbf{1},\mathbf{2})_{2, 1, 0, 1, -2, 2, 0, -1, 0, 0, 0}&1\times (\mathbf{1},\mathbf{1},\mathbf{2},\mathbf{1},\mathbf{1})_{-1, 1, 1, 0, 1, -2, 1, 1, 0, 1, 0} \\1\times (\mathbf{1},\mathbf{1},\mathbf{1},\mathbf{1},\mathbf{2})_{0, 1, 0, -2, 1, 0, 0, 0, 0, 2, 0}&1\times (\mathbf{1},\mathbf{1},\mathbf{1},\mathbf{1},\mathbf{2})_{-1, 1, 0, 0, 0, -2, 0, 1, 0, 0, 0}& 1\times (\mathbf{1},\mathbf{1},\mathbf{2},\mathbf{1},\mathbf{2})_{0, -1, 1, 1, 1, 0, 1, 0, 0, 0, 0}\\1\times (\mathbf{1},\mathbf{2},\mathbf{1},\mathbf{1},\mathbf{1})_{0, 0, 0, 0, 0, 0, 3, 1, 1, 0, -1}&1\times (\mathbf{1},\mathbf{1},\mathbf{2},\mathbf{1},\mathbf{1})_{0, 0, -1, 0, -1, 1, -1, -1, 1, 0, -2}&1\times (\mathbf{1},\mathbf{1},\mathbf{2},\mathbf{1},\mathbf{1})_{0, 0, -1, 0, -1, -1, 3, 1, -1, 0, 2}\\2\times (\mathbf{1},\mathbf{1},\mathbf{1},\mathbf{1},\mathbf{1})_{1, 0, -2, -1, -1, 3, -1, -1, 0, -1, 0} &1\times (\mathbf{2},\mathbf{2},\mathbf{1},\mathbf{1},\mathbf{1})_{0, 0, 2, 1, 0, -2, 3, 1, 1, 1, 0}&2\times (\mathbf{1},\mathbf{1},\mathbf{2},\mathbf{1},\mathbf{1})_{1, 0, 1, 0, 0, 1, 1, 0, 1, 0, 0} \\1\times (\mathbf{2},\mathbf{2},\mathbf{2},\mathbf{1},\mathbf{1})_{0, 0, -1, 0, -1, 0, 1, 0, 0, 0, 0}&1\times (\mathbf{2},\mathbf{2},\mathbf{1},\mathbf{1},\mathbf{1})_{0, 0, 0, 0, 0, 0, -2, -1, 0, 0, 0}&1\times (\mathbf{1},\mathbf{1},\mathbf{1},\mathbf{1},\mathbf{1})_{1, -1, 1, 2, -1, 2, -3, -2, 1, -1, -1}\\1\times (\mathbf{2},\mathbf{1},\mathbf{1},\mathbf{1},\mathbf{1})_{1, -1, 1, 2, -1, 1, 2, 0, 1, -1, 0}&1\times (\mathbf{2},\mathbf{1},\mathbf{1},\mathbf{1},\mathbf{1})_{-1, 0, 1, -1, 1, -3, 2, 2, 1, 0, 0}&1\times (\mathbf{1},\mathbf{1},\mathbf{1},\mathbf{1},\mathbf{1})_{0, -1, 0, 0, 0, 0, -1, 0, 1, -2, 0}\\1\times (\mathbf{2},\mathbf{1},\mathbf{1},\mathbf{1},\mathbf{1})_{0, 0, -1, -1, 0, 0, 0, 1, 0, -2, 0}&1\times (\mathbf{1},\mathbf{1},\mathbf{1},\mathbf{1},\mathbf{1})_{0, 0, 1, 0, 0, -2, 0, 1, 0, -1, 1}&1\times (\mathbf{1},\mathbf{1},\mathbf{1},\mathbf{1},\mathbf{1})_{0, 2, 1, -1, 1, -1, 0, 0, -1, 3, 1}\\1\times (\mathbf{1},\mathbf{1},\mathbf{1},\mathbf{1},\mathbf{1})_{1, -1, 1, 2, 0, 1, 0, -1, -1, 0, 1}&2\times (\mathbf{1},\mathbf{1},\mathbf{1},\mathbf{1},\mathbf{1})_{0, -1, 1, 0, 1, 0, 0, 0, 0, 0, 1}&2\times (\mathbf{1},\mathbf{1},\mathbf{2},\mathbf{1},\mathbf{1})_{0, 0, 0, 0, -1, -1, -1, 0, 1, -1, -1}\\2\times (\mathbf{1},\mathbf{1},\mathbf{1},\mathbf{1},\mathbf{1})_{0, 0, 0, -1, 1, 0, 2, 1, 0, 0, 0}&2\times (\mathbf{1},\mathbf{1},\mathbf{1},\mathbf{1},\mathbf{1})_{1, 1, 0, 0, 0, 2, 2, 0, 0, 1, 0}&1\times (\mathbf{1},\mathbf{1},\mathbf{1},\mathbf{2},\mathbf{1})_{-1, 0, 0, -1, 1, -1, 2, 2, 1, -1, 0}\\1\times (\mathbf{1},\mathbf{1},\mathbf{1},\mathbf{1},\mathbf{1})_{1, -1, 0, 1, -1, 1, 2, 1, 1, -3, 0}&1\times (\mathbf{1},\mathbf{1},\mathbf{1},\mathbf{1},\mathbf{2})_{-1, 1, 0, -1, 1, -2, 2, 2, 0, 0, 0}&1\times (\mathbf{1},\mathbf{1},\mathbf{1},\mathbf{1},\mathbf{1})_{0, -1, 0, 0, 1, 0, 2, 1, 0, -1, 0}\\1\times (\mathbf{1},\mathbf{1},\mathbf{1},\mathbf{1},\mathbf{1})_{0, 0, 0, 0, -1, -1, -1, -1, 0, 1, 0}&1\times (\mathbf{1},\mathbf{1},\mathbf{1},\mathbf{1},\mathbf{1})_{1, 1, 0, 1, -2, 1, -1, -2, 0, 2, 0}&1\times (\mathbf{1},\mathbf{1},\mathbf{1},\mathbf{1},\mathbf{1})_{0, -1, 0, 1, -1, -1, -1, -1, 0, 0, 0}\\1\times (\mathbf{1},\mathbf{1},\mathbf{1},\mathbf{1},\mathbf{1})_{1, -1, 1, 1, 0, 2, -1, -1, 1, -1, -1}&1\times (\mathbf{1},\mathbf{1},\mathbf{1},\mathbf{1},\mathbf{1})_{-2, -1, 1, 0, 2, -2, -1, 1, 1, -1, -1}&1\times (\mathbf{1},\mathbf{1},\mathbf{1},\mathbf{1},\mathbf{1})_{0, 0, 1, 0, 1, -1, -1, 0, 0, 0, -1}\\1\times (\mathbf{1},\mathbf{1},\mathbf{1},\mathbf{1},\mathbf{1})_{0, 1, 1, -1, 1, 0, -1, 0, 1, 1, -1}&1\times (\mathbf{1},\mathbf{1},\mathbf{2},\mathbf{1},\mathbf{1})_{0, 1, 0, -1, 0, -1, 2, 1, 0, 1, 1}&2\times (\mathbf{1},\mathbf{1},\mathbf{1},\mathbf{1},\mathbf{2})_{0, 0, 0, 1, -1, 0, 2, 0, 0, 1, 0}\\2\times (\mathbf{1},\mathbf{1},\mathbf{1},\mathbf{2},\mathbf{1})_{-1, 0, 0, 0, 0, -2, 2, 1, 0, 1, 0}&2\times (\mathbf{1},\mathbf{1},\mathbf{1},\mathbf{1},\mathbf{1})_{1, 0, 0, 1, -2, 1, 2, 0, 1, 0, 0}&3\times (\mathbf{1},\mathbf{1},\mathbf{1},\mathbf{1},\mathbf{1})_{0, -1, 0, 0, -1, -1, 2, 1, 1, -1, 0}\\3\times (\mathbf{1},\mathbf{1},\mathbf{1},\mathbf{1},\mathbf{1})_{1, -1, 0, 2, -2, 0, 2, 0, 0, -1, 0}&4\times (\mathbf{1},\mathbf{1},\mathbf{1},\mathbf{1},\mathbf{1})_{0, 2, 0, -1, -1, -2, 2, 1, 0, 2, 0}&4\times (\mathbf{1},\mathbf{1},\mathbf{1},\mathbf{2},\mathbf{1})_{0, 0, 0, 1, -1, -1, 1, 0, 0, 0, 0}\\4\times (\mathbf{1},\mathbf{1},\mathbf{1},\mathbf{1},\mathbf{2})_{0, -1, 0, 1, -1, 0, 1, 0, 1, -1, 0}&4\times (\mathbf{1},\mathbf{1},\mathbf{1},\mathbf{1},\mathbf{1})_{-2, 1, 0, -1, 1, -3, 1, 1, 0, 2, 0}&3\times (\mathbf{1},\mathbf{1},\mathbf{1},\mathbf{1},\mathbf{1})_{1, 1, 0, 0, -1, 1, 1, -1, 0, 2, 0}\\3\times (\mathbf{1},\mathbf{1},\mathbf{1},\mathbf{1},\mathbf{1})_{0, -1, 0, 0, 0, -1, 1, 0, 0, 0, 0}&3\times (\mathbf{1},\mathbf{1},\mathbf{1},\mathbf{1},\mathbf{1})_{0, 1, 0, -1, -1, -2, 1, 1, 1, 0, 0}&3\times (\mathbf{1},\mathbf{1},\mathbf{1},\mathbf{1},\mathbf{1})_{1, 0, 0, 1, -1, 1, 1, -1, 0, 1, 0}\\1\times (\mathbf{2},\mathbf{1},\mathbf{1},\mathbf{1},\mathbf{1})_{-1, 1, 0, 0, 0, -1, -1, 0, 0, 1, -1}&1\times (\mathbf{2},\mathbf{1},\mathbf{1},\mathbf{1},\mathbf{1})_{1, 0, 0, 0, -1, 2, -1, -1, 1, 0, -1}&1\times (\mathbf{1},\mathbf{2},\mathbf{1},\mathbf{1},\mathbf{1})_{1, 0, 0, 0, -1, 1, 1, 0, 0, 0, 1}\\1\times (\mathbf{1},\mathbf{1},\mathbf{1},\mathbf{1},\mathbf{1})_{1, 0, 0, 1, -3, 0, 1, -1, 1, 1, 0}&1\times (\mathbf{1},\mathbf{1},\mathbf{1},\mathbf{1},\mathbf{1})_{-1, 1, -2, -1, 0, 0, -1, 0, -1, 0, 0}&\\
\hline
    \end{array}$ }
\end{align*}
\caption{Charged matter for model 4.}
\label{tab:4d-model4}
\end{table}

\bibliographystyle{utphys} 
\bibliography{refs}

\end{document}